\documentclass[sigconf]{acmart}

\usepackage{color}
\usepackage[vlined,linesnumbered,ruled,norelsize]{algorithm2e}
\usepackage{amssymb}
\usepackage{amsfonts}
\usepackage[switch]{lineno}
\def\header{\vspace{2.5mm} \noindent}

\newcommand{\E}{{\mathbb E}\xspace}

\newcommand{\KM}{\mathit{KM}}
\newcommand{\KC}{\mathit{KC}}
\newcommand{\UB}{\mathit{UB}}
\newcommand{\R}{\mathcal{R}}
\newcommand{\OPT}{\mathrm{OPT}}
\DeclareSymbolFont{largesymbol}{OMX}{yhex}{m}{n}
\DeclareMathAccent{\wht}{\mathord}{largesymbol}{"62}

\begin{document}
\title{Approximation Algorithms for Probabilistic Graphs}
%\titlenote{Produces the permission block, and
%  copyright information}
%\subtitle{Extended Abstract}
%\subtitlenote{The full version of the author's guide is available as
%  \texttt{acmart.pdf} document}

%\author{
%   Kai Han$^\star$~~~~~~~~Chaoting Xu$^\star$~~~~~~~~Xiaokui Xiao$^\dag$~~~~~~~~Fei Gui$^\star$~~~~~~~~He Huang$^\ddag$ \vspace{.8ex}\\
%   \affaddr{$^\star$SCST, University of Science and Technology of China, P.R.China, 230026} \vspace{.5ex}\\
%   \affaddr{$^\dag$SCSE, Nanyang Technological University, Singapore, 639798} \vspace{.5ex}\\
%   \affaddr{$^\ddag$SCST, Soochow University, P.R.China, 215000} \\
%%   \email{hankai@ustc.edu.cn, huangh@suda.edu.cn, junluo@ntu.edu.sg}
%}
%

\author{Kai Han}
%\authornote{Dr.~Trovato insisted his name be first.}
\affiliation{%
  \institution{School of Computer Science and Technology, University of Science and Technology of China, P.R.China}
}
\email{hankai@ustc.edu.cn}

\begin{abstract}
We study the $k$-median and $k$-center problems in uncertain graphs. We analyze the hardness of these problems, and propose several algorithms with improved approximation ratios compared with the existing proposals.
\end{abstract}

%\begin{CCSXML}
%<ccs2012>
%<concept>
%<concept_id>10003033.10003106.10003114.10011730</concept_id>
%<concept_desc>Networks~Online social networks</concept_desc>
%<concept_significance>500</concept_significance>
%</concept>
%</ccs2012>
%\end{CCSXML}
%
%\ccsdesc[500]{Networks~Online social networks}
%%
%%
%\keywords{Social networks, discount allocation, revenue}

\maketitle

\section{Introduction} \label{sec:intro}

%Say something about IM. Mentioning the following terms: advertiser, activation, influenced individual.
Graph data are prevalent in a lot of application domains such as social, biological and mobile networks. Typically, the entities in realities are modeled by graph nodes, and the relationships between entities are modeled by graph edges. Uncertainty is evident in graph data due to a variety of reasons.
 %such as noisy measurements, inconsistent, incorrect,
%and possibly ambiguous information sources, lack of precise information needs, inference and prediction models, or explicit manipulation~\cite{}. Therefore, data is usually represented as an uncertain graph, that is, a graph whose nodes, edges, and attributes are accompanied with a probability of existence.
Therefore, the methods for querying and mining uncertain graph data are of paramount importance.

Graph clustering is a fundamental problem in graph data mining, where the goal is to partition the graph nodes into some clusters, such that the nodes in each cluster is ``close'' to each other according to some distance measure. Among the numerous problem definitions on graph clustering, the $k$-median and $k$-center problems are perhaps the most celebrated ones which have been studied for decades~\cite{Vazirani2001}. In a traditional graph (without uncertainty), the goal of the $k$-median problem is to find $k$ centering nodes in the network such that the average distance between each node to the centering nodes is maximized, while the goal of the $k$-center problem is to find a set of $k$ nodes for which the largest distance of any point to its closest vertex in the $k$-set is minimum.

Surprisingly, although the $k$-median and $k$-center problems have been extensively studied in the literature, their counterpart problems in uncertain graphs have not been investigated until a recent study by Ceccarello et al.~\cite{Ceccarello17}. Following a large body of work on uncertain graphs, the work in~\cite{Ceccarello17} models an uncertain graph as a traditional graph augmented by existence probabilities associated to the edges. They use the connection probabilities as the distance measure between the nodes, and formulated the $k$-median and $k$-center problems as follows. In the $k$-median problem, they aim to partition the graph nodes into $k$ subsets (clusterings) with a centering node in each of them, such that the average connection probability between each node and its corresponding centering node is maximized. In the $k$-center problem, they aim to maximize the minimum connection connectivity between a node and its centering node. It can be seen that the definitions of their $k$-median and $k$-center problems are in spirit similar to those for the traditional graphs, so they can be considered as the reinterpretations of the $k$-median and $k$-center problems in traditional graphs.

In contrast to the traditional $k$-median and $k$-center problems, there are two unique challenges for clustering uncertain graphs. First, it is a \#P hard problem to compute the connection probability between any two nodes in an uncertain graph. Second, the distance measure described above does obey the triangle inequality, which is required by almost all of the traditional $k$-center and $k$-median algorithms. Therefore, even if we have an oracle for computing the connection probabilities, the traditional $k$-median and $k$-center algorithms cannot be applied to our case.

Based on the above observations, the work in~\cite{Ceccarello17} provide new algorithms for graph clustering problem in uncertain graphs. However, the approximation ratio of their algorithms are far from satisfactory.

\header
{\bf Contributions.} Motivated by the deficiency of existing techniques, we propose new approximation algorithms for the $k$-median and $k$-center problems in uncertain graphs. Our contributions are summarized as follows.

(1) For the $k$-median problem:

We prove that the $k$-median problem is NP-hard, and propose an approximate algorithm with a $1-1/e$ approximation ratio. We also propose efficient sampling algorithms that achieves a $1-1/e-\epsilon$ approximation ratio when there does not exist an oracle for computing the connection probabilities.

(2) For the $k$-center problem:

We prove that the $k$-center problem is NP-hard to approximate within any bounded ratio. We first propose a simple algorithm with the approximation ratio of $\OPT_k^c$, and then provide a bi-criteria approximation algorithm that achieve $1-\epsilon$ approximation ratio using at most $\mathcal{O}(k\log\frac{n}{\epsilon})$ centering nodes. We also propose algorithms for the $k$-center problem without the connection oracle.

\section{Preliminaries} \label{sec:prelim}

\subsection{Problem Definitions} \label{sec:prelim-IM}

An uncertain graph is represented by $G = (V, E)$ where $V$ is the set of nodes and $E$ is the set edges, with $|V|=n$ and $|E|=m$. We assume that each node in $V$ has a unique node ID in $[1,n]$. Each edge $e\in E$ is associated with a number $p(e)\in (0,1]$ denoting the probability that $e$ exists. For any two nodes $u$ and $v$ in $V$, we use $\mathrm{Pr}[u\sim v]$ to denote the probability that $u$ and $v$ is connected in $G$. For simplicity, we follow the work in~\cite{Ceccarello17} to assume that $G$ is an undirected graph, but our approach can be readily extended to the case of directed graphs, which will be explained later.

A \textit{$k$-clustering} of $G$ can be represented by a tuple $\mathcal{C}=\langle C, Q_1,Q_2,\cdots,Q_k\rangle$, where $C=\{c_1,\cdots,c_k\}$ is the set of \textit{centering nodes} and  $\{Q_1,Q_2,\cdots,Q_k\}$ is a partition of the nodes in $V$ satisfying $c_i\in Q_i$ for all $i\in [k]$. For any $i\in [k]$ and any $v\in Q_i$, we call the node pair $(c_i,v)$ as a \textit{cluster link} of $\mathcal{C}$. We call the set of all clustering links in $\mathcal{C}$ as the \textit{signature} of $\mathcal{C}$, and we use $\mathcal{S}_k^{G}$ to denote the set of signatures of all possible $k$-clusterings of $G$. Note that any two different $k$-clusterings must have different signatures, and we can construct a $k$-clustering from any $A\in \mathcal{S}_k^G$. Therefore, we will also call any $A\in \mathcal{S}_k^G$ as a $k$-clustering. Given any $A\in \mathcal{S}_k^G$, we define
\begin{eqnarray}
&&\mathit{KM}(A)={\sum\nolimits_{(u,v)\in A}{\mathrm{Pr}}[u\sim v]}/{n}\\
&&\mathit{KC}(A)=\min\nolimits_{(u,v)\in A}\mathrm{Pr}[u\sim v]
\end{eqnarray}
With the above definitions, the $k$-median and $k$-center problems can be formally defined as follows:

\begin{definition}
The $k$-median (KMD) problem aims to identify an optimal solution $A^o$ to the following optimization problem:
\begin{eqnarray}
        {\mathbf{Maximize}}&&~~\mathit{KM}(A) \qquad\qquad\qquad\qquad \mathbf{[KMD]}\nonumber\\
        \mathbf{s.t.}&&~~A\in \mathcal{S}_k^G \nonumber
\end{eqnarray}
The $k$-center (KCT) problem aims to identify an optimal solution $B^o$ to the following optimization problem:
\begin{eqnarray}
        {\mathbf{Maximize}}&&~~\mathit{KC}(B) \qquad\qquad\qquad\qquad \mathbf{[KCT]}\nonumber\\
        \mathbf{s.t.}&&~~B\in \mathcal{S}_k^G \nonumber
\end{eqnarray}
\end{definition}

%In other words, the $k$-median problem aims to find a $k$-clustering to maximize the average connection probabilities of all its cluster links, while the $k$-median problem aims to find a $k$-clustering to maximize the minimum connection probability of all its cluster links. It can be seen that these definitions are similar in spirit to the traditional $k$-median and $k$-center problems in graphs without uncertainty~\cite{}.

For convenience, we use $\mathrm{OPT}_{k}^m$ to denote $\mathit{KM}(A^o)$, and use $\mathrm{OPT}_{k}^c$ to denote $\mathit{KC}(B^o)$. %Table~\ref{tbl:prelim-notations} lists the notations that are frequently used in the remainder of the paper.

\section{Solving the $k$-median problem} \label{sec:frame}

%\subsection{Solutions for $k$-Median}
%
%and analyze the hardness of the $k$-median and $k$-center problems under this assumption. We also propose several new algorithms to address these two clustering problems. The main idea of our algorithms is that we can reformulate the $k$-median and $k$-center problems into some submodular optimization problems, and hence

\subsection{Hardness of the $k$-Median Problem}

The prior work~\cite{Ceccarello17} has conjectured that the $k$-median problem is NP-hard. We prove this conjecture in the following theorem, by a reduction from the NP-hard Dominating Set problem:

\begin{theorem}
The $k$-median problem is NP-hard, even if there exists an oracle for computing $\forall u,v\in V: \Pr[u\sim v]$.
\label{thm:hardnesskmedian}
\end{theorem}
\begin{proof}
We prove the theorem by a reduction from the NP-hard \textit{dominating set} problem~\cite{Vazirani2001}. Given any undirected graph $G=(V,E)$ with $|V|=n$ and any integer $k$, the decision version of the dominating set problem asks whether there exists $S\subseteq V$ with $|S|= k$ such that each node in $V\backslash S$ is adjacent to certain node in $S$. Given such an instance $G$ of the dominating set problem, we can construct an uncertain graph by setting $p(e)=q=\frac{1}{n(n-k+2)}$ for each $e\in E$. Suppose that there exists a polynomial-time algorithm $\mathcal{A}_{opt}$ to optimally solve the $k$-median problem. So we can run it on the uncertain graph $G$ described above, and get an optimal $k$-clustering with its signature denoted by $\tilde{A}$. In the sequel, we will prove that: the graph $G$ has a dominating set $S$ satisfying $|S|= k$ if and only if $\mathit{KM}(\tilde{A})\geq k+(n-k)q$

%\textit{The graph $G'$ has a dominating set $S$ satisfying $|S|= k$ if and only if $\mathit{KM}(A'')\geq k+(n-k)q$.}

If the graph $G$ has a dominating set $S$ satisfying $|S|= k$, then we can use $S$ as the set of center nodes, and hence we must have $\mathit{KM}(\tilde{A})\geq k+(n-k)q$. Conversely, if $G$ does not have a dominating set $S$ satisfying $|S|= k$, then there must exist a cluster link $(u,v)$ in $\tilde{A}$ such that $u$ and $v$ are not adjacent, and hence we get
\begin{eqnarray}
\mathrm{Pr}[u\sim v]\leq nq^2+ n^2q^3 + n^3 q^4 +\cdots \leq \frac{nq^2}{1-nq},
\end{eqnarray}
where $n^iq^{i+1}$ is an upper bound for the probability that $u$ is connected to $v$ through $i+1$ hops. Moreover, for any $(u',v')\in \tilde{A}\backslash \{(u,v)\}$ satisfying $u'\neq v'$, we must have
\begin{eqnarray}
\mathrm{Pr}[u'\sim v']\leq q+{nq^2}/{(1-nq)}
\end{eqnarray}
Therefore, we have
\begin{eqnarray}
\mathit{KM}(\tilde{A})&\leq& k+ (n-k-1)\left(q+\frac{nq^2}{1-nq}\right)+\frac{nq^2}{1-nq}\nonumber\\
&<& k+(n-k)q
\end{eqnarray}
The above reasoning implies that, if $\mathcal{A}_{opt}$ exists, then the dominating set problem can also be optimally solved in polynomial time. Hence, the theorem follows.
\end{proof}

\subsection{$k$-Median Algorithms with an Oracle}

In this section, we assume that there exists a \textit{connectivity oracle}, i.e., $\mathrm{Pr}[u\sim v]$ can be computed in polynomial time for any $u\in V$ and $v\in V$.

It is highly non-trivial to find an approximation solution to the $k$-median problem, as it has a large searching space $\mathcal{S}_k^G$ with the cardinality of ${n\choose k}k^{n-k}$. However, we find that the $k$-median problem can be transformed into a submodular maximization problem with a much-reduced searching space, as described below.

For any $C\subseteq V$ and any $v\in V$, we define
\begin{eqnarray}
f_v(C)=\max\{\mathrm{Pr}[u\sim v]\mid u\in C\};~F(C)=\sum\limits_{v\in V}f_v(C) %~\mathit{KC}(C)=\mathit{KM}(D(C))
\end{eqnarray}

\begin{algorithm} [t]
%    \KwIn{\textcolor{black}{$\gamma, \epsilon$; $g(\cdot)$ is a submodular function}}%where $0<\sigma+\gamma<\alpha$}
%    \KwOut{A set $S\subseteq V$ and an interval $[a,b]\subseteq [1,n]$}
        $C\leftarrow \emptyset$\\
        \While{$|C|<k$}{
            Find $u\in V\backslash C$ such that $g(C\cup\{u\})-g(C)$ is maximized;\\
            $C\leftarrow C\cup\{u\}$
        }
  \Return{$C$}
  \caption{$\mathsf{Greedy}(G, k,g(\cdot))$}
    \label{alg:greedyalgorithm}
\end{algorithm}

It is noted that, for any $A\in \mathcal{S}_k^G$, there must exist certain $C\subseteq V$ such that $|C|=k$ and $F(C)\geq \mathit{KM}(A)$. Moreover, given any $C\subseteq V$, we can easily construct a $k$-clustering $A$ such that $C$ is the set of centering nodes in $A$ and $F(C)=\KM(A)$.  Therefore, the $k$-Median problem can be transformed into the following equivalent optimization problem:
  \begin{eqnarray}
        {\mathbf{Maximize}}&&~~\sum\nolimits_{v\in V}f_v(C) \qquad\qquad\qquad\qquad \mathbf{[KMD1]}\nonumber\\
        \mathbf{s.t.}&&~~|C|= k;~~C\subseteq V \nonumber
  \end{eqnarray}
Moreover, we find that the [KMD1] problem is actually a submodular maximization problem, as shown by the following theorem:

\begin{theorem}
For any $v\in V$, the function $f_v(\cdot)$ is a monotone and submodular function defined on $2^V$.
\label{thm:submodularitythm}
\end{theorem}
It is a well-known fact that monotone submodular maximization problems can be addressed by a greedy algorithm with a $1-1/e$ approximation ratio. Therefore, we can use a greedy algorithm (shown in Algorithm~\ref{alg:greedyalgorithm}) to find a $1-1/e$ approximation to [KMD1] and hence to the $k$-median problem. i.e.:

\begin{theorem}
Using the $\mathsf{Greedy}(G, k,F(\cdot))$ algorithm, we can find a solution with a $1-1/e$ approximation ratio ti the $k$-median problem.
\end{theorem}

Besides, we can also provide a $1/k$ approximation to the $k$-median problem. Consider the following problem:

\begin{eqnarray}
        {\mathbf{Maximize}}&&~~\sum_{u\in C}\sum_{v\in V}\Pr[u\sim v] \qquad\qquad\qquad \mathbf{[KMD2]}\nonumber\\
        \mathbf{s.t.}&&~~|C|= k;~~C\subseteq V \nonumber
\end{eqnarray}

We have the following theorem:

\begin{theorem}
Suppose that $C^{\dag}$ is an optimal solution to the [KMD2] problem. We have $F(C^{\dag})\geq (1/k)\KM(A^o)$.
\end{theorem}
\begin{proof}
Suppose that $C^o_{km}$ is the set of centering nodes in $A^o$. We have
\begin{eqnarray}
&&F(C^{\dag})=\sum\nolimits_{v\in V}\max\{\mathrm{Pr}[u\sim v]\mid u\in C^{\dag}\} \nonumber\\
&\geq& \sum\nolimits_{v\in V}(1/k)\sum\nolimits_{u\in C^{\dag}} \Pr[u\sim v] \nonumber\\
&\geq& (1/k)\sum\nolimits_{v\in V}\sum\nolimits_{u\in C^o_{km}} \Pr[u\sim v] \nonumber\\
&\geq& (1/k)\KM(A^o) \nonumber
\end{eqnarray}
Hence, the theorem follows.
\end{proof}

\subsection{$k$-Median Algorithms without Oracle}

In this section, we consider a more practical setting where the connection oracle is absent. We will first provide a basic sampling algorithm to address the $k$-median problem, and then provide some more efficient algorithms.

\subsubsection{A Basic Sampling Algorithm for $k$-Median}

A random sample $R$ of $G$ is a graph generated by removing each edge $e$ in $G$ with the probability of $1-p(e)$. For any $u,v\in V$ and any random sample $R$ of $G$, let $X_R(u\sim v)=1$ when $u$ and $v$ is connected in $R$, and $X_R(u\sim v)=0$ when $u$ and $v$ is not connected in $R$. For any set $\mathcal{R}$ of random samples of $G$, define
\begin{eqnarray}
\wht{\mathrm{Pr}}}[{\mathcal{R},u\sim v]=\sum\nolimits_{R\in \mathcal{R}}X_R(u\sim v)/|\mathcal{R}|\nonumber\\
{\wht{\mathit{KM}}}({\mathcal{R}},A)={\sum\nolimits_{(u,v)\in A}\wht{\mathrm{Pr}}[{\mathcal{R}}, u\sim v]}/{n}\nonumber
\end{eqnarray}
It can be seen that $\wht{\mathrm{Pr}}[{\mathcal{R}},u\sim v]$ and $\wht{\KM}(\mathcal{R},A)$ are unbiased estimations of ${\mathrm{Pr}}[u\sim v]$ and $\KM(A)$, respectively. Similarly, for any $A\in \mathcal{S}_k$, any $C\subseteq V$ and any $v\in V$, we define
\begin{eqnarray}
&&\wht{f}_v(\mathcal{R},C)=\max\{\wht{\mathrm{Pr}}}[{\mathcal{R},u\sim v]\mid u\in C\}\nonumber\\
&&\wht{F}(\R,C)=\sum\nolimits_{v\in V}\wht{f}_v(\mathcal{R},C)\nonumber
\end{eqnarray}

%Given any node set $C\subseteq V$ and any set $\mathcal{R}$ of random samples of $G$, we can find a clustering of $G$ as follows. Suppose that the nodes in $C$ are $c_1,c_2,\cdots,c_l$. We regard the nodes in $C$ as the centers, and add each node $v\in V\backslash C$ into the cluster that $c_i$ lies in, such that the node ID of $c_i$ is minimized and $\wht{\mathrm{Pr}}[\mathcal{R},c_i\sim v]=\wht{f}_v(\mathcal{R},C)$. We use $\wht{D}(\mathcal{R},C)$ to denote the signature of such a clustering generated from $C$ and $\mathcal{R}$.

\begin{algorithm} [t]
%    \KwIn{\textcolor{black}{$\gamma, \epsilon$; $g(\cdot)$ is a submodular function}}%where $0<\sigma+\gamma<\alpha$}
%    \KwOut{A set $S\subseteq V$ and an interval $[a,b]\subseteq [1,n]$}
  $C^*\leftarrow \mathsf{Greedy}(G,k,\wht{F}(\R,\cdot));~~ A^*\leftarrow D(\mathcal{R},C^*)$\\
  \Return{$A^*$}
    \caption{$\mathsf{SearchKM}(G,k,\mathcal{R})$}
    \label{alg:searchkm}
\end{algorithm}

% For convenience, we abuse the notations a little by defining
%\begin{eqnarray}
%&&\forall C\subseteq V: \wht{\mathit{KM}}({\mathcal{R}},C)=\wht{\mathit{KM}}({\mathcal{R}},\wht{D}(\mathcal{R},C))\\
%&&\forall C\subseteq V: \wht{\mathit{KC}}({\mathcal{R}},C)=\wht{\mathit{KC}}({\mathcal{R}},\wht{D}(\mathcal{R},C))
%\end{eqnarray}

As computing the connection probabilities is NP-hard, it is hard to find the cluster links even if we know the set of centering nodes in an optimal solution. To bypass this difficulty, we create a mapping between $V_k=\{C|C\subseteq V\wedge |C|=k\}$ and $\mathcal{S}_k^G$ to reduce the number of generated samples for identifying the cluster links. More specifically, given any set $\R$ of random samples, each node set $C\in V_k$ is mapped to a unique $k$-clustering $D(\R,C)$ in $\mathcal{S}_k^G$, such that any cluster link $(c,v)\in D(\R,C)$ (with $c\in C$) that satisfies: 1) $\wht{\mathrm{Pr}}[\mathcal{R},c\sim v]=\wht{f}_v(\mathcal{R},C)$; 2) The node ID of $c$ is minimized under condition 1).

With the above definitions, we design the $\mathsf{SearchKM}$ algorithm to find an approximate solution. Given a set $\R$ of random samples, the $\mathsf{SearchKM}$ algorithm first calls the $\mathsf{Greedy}$ algorithm to find a set $C^*$ of center nodes, and then returns $A^*=D(\R,C^*)$ as an approximate solution. To ensure that $A^*$ has a good approximation ratio, we give the following theorem to determine a upper-bound for the number of generated samples:
\begin{theorem}
If $|\mathcal{R}|\geq T_{max}$ where
\begin{eqnarray}
T_{max}= \left\lceil \frac{2(2e-1)(e\epsilon+2e-1)}{3e^2\epsilon^2\mathrm{OPT}_{k}^m}\ln\frac{{n\choose k}+1}{\delta}\right\rceil=\mathcal{O}\left(\frac{n}{\epsilon^2}\log\frac{n}{\delta}\right),\nonumber
\end{eqnarray}
then the $\mathsf{SearchKM}(G,k,\mathcal{R})$ algorithm returns a $(1-1/e-\epsilon)$-approximate solution $A^*$ to the $k$-median problem with probability of at least $1-\delta$.
\label{thm:samplingubforkmedian}
\end{theorem}

%Note that the work in~\cite{} requires at least $\mathcal{O}(\frac{n^3}{\epsilon^2 k^3})\log\frac{n}{\delta}$ number of samples. First, the $\mathsf{SearchKM}$

\subsubsection{Accelerations}

The $\mathsf{SearchKM}$ algorithm can be further accelerated by leveraging the submodularity of the function $\wht{F}(\R,\cdot)$, as shown by the $\mathsf{SearchKM+}$ algorithm. The $\mathsf{SearchKM+}$ algorithm maintains a value $\UB(v)$ for each $v\in V$, which denotes an upper bound for the marginal gain of $v$ with respect to the currently selected node set $C^*$. Initially, $\mathsf{SearchKM+}$ calls $\mathsf{GetFirstNode}(\R)$ to calculate $\UB(v)=\wht{F}(\R,\{v\})$ for all $v\in V$, and then add $u^*=\arg\max_{u\in V}\UB(v)$ into $C^*$. After that, it sorts $V$ into the node list $W$ according to the non-increasing order of $\forall v\in V: \UB(v)$, and re-compute $\UB(v)$ only when necessary.

It can be seen that the idea of $\mathsf{SearchKM+}$ is similar in spirit to the ``lazy greedy'' algorithm proposed in~\cite{LeskovecKGFVG2007}. However, the lazy greedy algorithm has not considered the ``cold start'' problem, i.e., how to efficiently compute the upper bound of the marginal gain of any node in $V$ (i.e., $\UB(v)$) in the initialization phase. In our case, a naive approach for computing the initial value of $\UB(v)$ requires $\mathcal{O}(n)$ time for any $v\in V$, as we need to calculate $\wht{f}_u(\R,\{v\})$ for each $u\in V$. However, using the $\mathsf{GetFirstNode}$ procedure, we only need $\mathcal{O}(1)$ time to compute $\UB(v)$.

%\begin{lemma}
%The $\mathsf{GetFirstNode}$ procedure can correctly calculate
%\end{lemma}

The $\mathsf{SearchKM+}$ algorithm can be further accelerated by borrowing some ideas from the OPIM sampling framework proposed in~\cite{Tang2018}. The resulted algorithm is shown in Algorithm~\ref{alg:opimkm}. We can prove:

%We roughly explain the idea as follows. Instead of generating all $T_{max}$ random samples of $G$ in one batch, we can progressively generate  As this process is similar to

\begin{theorem}
With probability of at least $1-\delta$, the $\mathsf{SamplingKM}$ algorithm returns a $k$-clustering with a $1-1/e-\epsilon$ approximation ratio. The expected number of random samples generated in $\mathsf{SamplingKM}$ is at most $\mathcal{O}(\frac{1}{\epsilon^2\OPT_k^m}\ln\frac{1}{\delta})$.
\end{theorem}

\begin{algorithm} [t]
%    \KwIn{\textcolor{black}{$\gamma, \epsilon$; $g(\cdot)$ is a submodular function}}%where $0<\sigma+\gamma<\alpha$}
%    \KwOut{A set $S\subseteq V$ and an interval $[a,b]\subseteq [1,n]$}
        $(u,W)\leftarrow \mathsf{GetFirstNode}(\mathcal{R});~~C^*\leftarrow \{u\}$\\
        \While{$|C^*|<k$}{
            $(u,W)\leftarrow \mathsf{GetNextNode}(W,\wht{F}(\mathcal{R},\cdot),C^*)$\\
            $C^*\leftarrow C^*\cup\{u\}$
        }
  \Return{$C^*$}
    \caption{$\mathsf{SearchKM+}(G,k,\mathcal{R})$}
    \label{alg:searchkm}
\end{algorithm}

\begin{algorithm} [t]
%   \KwIn{$\epsilon, \epsilon_1,\epsilon_2$; $g(\cdot)$ is a submodular function}%where $0<\sigma+\gamma<\alpha$}
%   \KwOut{A set $S\subseteq V$ and an interval $[a,b]\subseteq [1,n]$}
        \ForEach{$v\in V$}{
            $\mathit{UB}(v)\leftarrow$ the summation of the sizes of the connected components in $\mathcal{R}$ that contain $v$;
		}
        Sort $V$ into the node list $W$ according to the non-increasing order of $\mathit{UB}(v):v\in W$;\\
      Remove the first node $w_1$ from $W$;\\
    \Return{$(w_1,W)$}
    \caption{$\mathsf{GetFirstNode}(\mathcal{R})$}
    \label{alg:bisearch}
\end{algorithm}

\begin{algorithm} [t!]
%   \KwIn{$\epsilon, \epsilon_1,\epsilon_2$; $g(\cdot)$ is a submodular function}%where $0<\sigma+\gamma<\alpha$}
%   \KwOut{A set $S\subseteq V$ and an interval $[a,b]\subseteq [1,n]$}
        %$a_{max}\leftarrow -1$\\
%        $l\leftarrow |W|$;\\
        %$l\leftarrow |W|$\\
        \For{$i\leftarrow 1$ \KwTo $|W|$}{
            $\mathit{UB}(w_i)\leftarrow g(C\cup \{w_i\})-g(C)$\\
            \lIf{$\mathit{UB}(w_i)\geq \mathit{UB}(w_{i+1})$}{\textbf{break}}
		}
    Re-sort the nodes in $W$ according to the non-increasing order of $\mathit{UB}(v): v\in W$\\
    Remove the first node $w_1$ from $W$;\\
    \Return{$(w_1,W)$}
    \caption{$\mathsf{GetNextNode}(W,g(\cdot),C)$}
    \label{alg:bisearch}
\end{algorithm}

%
%\section{A Batched Sampling Algorithm} \label{sec:batch}
%
%In this section, we propose an efficient sampling algorithm for the selecting the seed nodes in each batch. We will first introduce the motivation and main idea of our algorithm, and then present the details of our algorithm.

\begin{algorithm}[ht!]
     \label{alg:opimkm}
	\caption{$\mathsf{SamplingKM}(G,k,\epsilon,\delta)$}
%	 \KwIn{Graph $G$, seed set size $k$, error threshold $\epsilon$, and failure probability threshold $\delta$}
%	 \KwOut{A $k$-size seed set $S^{*}$ that provides an approximation guarantee of at least $(1 - 1/e - \epsilon)$ with probability at least 1 - $\delta$ }
    $T_{max}\leftarrow \frac{2(2e-1)(e\epsilon+2e-1)n}{3e^2\epsilon^2 k}\ln\frac{{n\choose k}+1}{\delta};~~T\leftarrow T\cdot \epsilon^2 k/n$\\
	Generate two sets $\mathcal{R}_{1}$ and $\mathcal{R}_{2}$ of random samples of $G$, such that $|\mathcal{R}_{1}|$ = $|\mathcal{R}_{2}|$ = $T$\;
	$i_{max} \gets \lceil \log_{2} (T_{max}/T) \rceil $\;
	\For{$i \gets 1$ $\mathbf{to}$ $i_{max}$}{
		            $A^*\leftarrow \mathsf{SearchKM}(G,k,\mathcal{R}_1)$ \\
                    $a\gets \ln (3 i_{max}/\delta);~~\theta\leftarrow |\mathcal{R}_1|$\\
	 				$\mathrm{lb}(A^*)\gets \left(\sqrt{\wht{\KM}(\mathcal{R}_2, A^*)} - \sqrt{\frac{a}{6\theta}}\right)^{2} - \frac{a}{6\theta}$\\
                    $\mathrm{ub}(A^o)\gets  \left( \sqrt{\frac{\wht{\KM}(\mathcal{R}_1, A^*)}{1 - 1/e} + \frac{2a}{3\theta}} + \sqrt{\frac{a}{6\theta}} \right)^2 - \frac{a}{6\theta}$\\
%                    compute $\sigma^{l}(S^{*})$ and $\sigma^{u}(S^{o})$ by $(3)$ and $(4)$,respectively, setting $\delta_{1} = \delta_{2} = \delta/(3i_{max})$\;
%	  				$\alpha \gets \sigma^{l}(S^{*})/\sigma^{u}(S^{o})$\;
	  				\If{ $\mathrm{lb}(A^*)/\mathrm{ub}(A^o) \ge 1-1/e-\epsilon $ $\mathbf{or}$ $ i=i_{max}$}{$\mathbf{return}$ $A^{*}$}
	     			double the sizes of $\mathcal{R}_{1}$ and $\mathcal{R}_{2}$ with new random samples;
     	}
\end{algorithm}

\section{Solving the k-center problem}

In this section, we address the $k$-center problem both with a connection oracle assumption and without it.

\subsection{k-Center Algorithms with an Oracle}

\subsubsection{A Simple Algorithm}
\label{sec:asimplealgorithm}

The work in~\cite{Ceccarello17} has proved that the connection probabilities of any three nodes $u,v,w\in V$ must satisfy
\begin{eqnarray}
\Pr[u\sim w]\geq \Pr[u\sim v]\cdot \Pr[v\sim w]
\end{eqnarray}
Let $d(u,v)=-\ln \Pr[u\sim v]$ and $d_v(C)=\min_{u\in C}d(u,v)$, so we have
\begin{eqnarray}
d(u, w)\leq  d(u, v)+ d(v,w)
\end{eqnarray}
which implies that $d(\cdot)$ is a metric. Consider the following problem:
  \begin{eqnarray}
        {\mathbf{Minimize}}&&~~\max_{v\in V} d_v(C) \qquad\qquad\qquad\qquad \mathbf{[KCT0]}\nonumber\\
        \mathbf{s.t.}&&~~|C|= k;~~C\subseteq V \nonumber
  \end{eqnarray}
It can be seen that the set of centering nodes in $B^o$ is also an optimal solution to the [KCT0] problem. Note that [KCT0] is a metric $k$-center problem, so it can be addressed by a simple greedy algorithm with a 2 approximation ratio~\cite{williamson2011design}. More specifically, the greedy algorithm initializes by selecting an arbitrary node, and then iteratively selects a node which is furthest to the currently selected nodes until $k$ nodes are selected. With this greedy algorithm, we can find $B^*\in \mathcal{S}_k^G$ such that $-\ln \KC(B^*)\leq -2\ln \KC(B^o)$, which implies that
\begin{eqnarray}
\KC(B^*)\geq (\OPT_k^c)^2  \label{eqn:arofgreedykcenter}
\end{eqnarray}

\subsubsection{A Bi-Criteria Approximation Algorithm}
\label{sec:bisearchwithoracle}

Note that the approximation ratio proposed by~\eqref{eqn:arofgreedykcenter} can be arbitrarily bad, as $\OPT_k^c$ can be arbitrarily small. Therefore, we ask whether there exists an algorithm with a bounded approximation ratio for the $k$-center problem. Unfortunately, we find that:

%The prior work in~\cite{Ceccarello17} has proved the $k$-center problem is NP-hard when there exists a connectivity oracle. We make a step further to prove that it is even inapproximable
%within any ratio:

\begin{theorem}
Unless P=NP, no polynomial-time algorithm can find a solution to the $k$-center problem within any approximation ratio $\alpha>0$, even if there exists a connectivity oracle.
\label{thm:inapprofkcenter}
\end{theorem}

As the $k$-center problem is NP-hard to approximate, we further ask the question whether there exists a \textit{bi-criteria approximation algorithm} for it, i.e., we permit such an algorithm to use more than $k$ center nodes, such that it can approach $\OPT_k^c$.
%the connection probabilities of its cluster links can be optimized to the maximum extent.
However, the following theorem reveals that, we cannot achieve a large connectivity probability unless we allow the usage of a ``sufficiently large'' number of centering nodes:

%that the number of the selected center nodes to be larger than $k$, such that we can find an approximate solution to the $k$-center problem with its connectivity performance close to $\mathrm{OPT}_{kc}$. However, the following theorem reveals that we cannot select ``too few'' center nodes to achieve such a goal:

\begin{theorem}
Unless P=NP, no algorithm can find a $l$-clustering $B$ in polynomial time, such that $\mathit{KC}(B)\geq \mathrm{OPT}_{k}^c$ and $l<k\ln n$.
\label{thm:kcboundedbylnn}
\end{theorem}

%As the $k$-center problem is NP-hard to approximate, we try to find a bi-criteria approximation algorithm for it.

Based on Theorem~\ref{thm:kcboundedbylnn}, we propose a bi-criteria approximation algorithm with nearly tight approximation ratios. First, we re-formulate the [KCT] problem into the following [KCT1] problem:
  \begin{eqnarray}
        {\mathbf{Maximize}}&&~~\min_{v\in V} f_v(C) \qquad\qquad\qquad\qquad \mathbf{[KCT1]}\nonumber\\
        \mathbf{s.t.}&&~~|C|= k;~~C\subseteq V \nonumber
  \end{eqnarray}
It can be seen that, for any $k$-clustering $B\in \mathcal{S}_k^G$, we must have $\KC(B)\leq \min_{v\in V} f_v(C_B)$, where $C_B$ denotes the set of centering nodes in $B$. Therefore, the [KCT1] problem is equivalent to the [KCT] problem.

\begin{algorithm} [ht!]
%   \KwIn{$\epsilon, \epsilon_1,\epsilon_2$; $g(\cdot)$ is a submodular function}%where $0<\sigma+\gamma<\alpha$}
%   \KwOut{A set $S\subseteq V$ and an interval $[a,b]\subseteq [1,n]$}
		$[q_1,q_2]\leftarrow [0,1]; C^*\leftarrow \emptyset$ \label{ln:setab}\\
		\Repeat{$q_1\geq(1-\epsilon_2)q_2$}
		{
			 $C\leftarrow \emptyset;~~q\leftarrow \frac{q_1+q_2}{2}$\\
            $L(q,\cdot)\gets \sum_{v\in V}\min\{q,f_v(\cdot)\}$\\
            \Repeat{$L(q,C)\geq nq-\epsilon_1 q$}
                {
                    $v^*\leftarrow \arg\max_{u\in V\backslash C}[L(q,C\cup \{u\})-L(q,C)]$;\\
                    $C\leftarrow C\cup \{v^*\}$\\
                    \lIf{$|C|>\lceil\ln\frac{n}{\epsilon_1 } \rceil k$}{\textbf{break}}
                }
            %Find $C$ such that $L(q,C)\geq nq-\epsilon_1$ and $|C|$ is (approximately) minimized.\label{ln:callgrecoverinbisearch}\\
			\eIf{$|C| \leq \lceil\ln\frac{n}{\epsilon_1 } \rceil k$ \label{ln:judgeca}}
			{
				$C^*\leftarrow C;~~ q_1 \leftarrow q$ \label{ln:setS}
			}
			{
				$q_2 \leftarrow q$ \label{ln:setb}
			}
            %$q \leftarrow (q_1 + q_2)/2$
		}\label{ln:repeatends}
  \Return{$C^*$}
    \caption{$\mathsf{SearchKC}(\epsilon_1,\epsilon_2, f_v(\cdot))$}
    \label{alg:bisearch}
  \end{algorithm}

Recall that the $f_v(\cdot)$ is monotone and submodular for any $v\in V$. Therefore, the [KCT1] problem is similar to the ``robust submodular maximization'' problem studied in~\cite{krause2008robust}. However, the algorithms and performance bounds proposed in~\cite{krause2008robust} are only suitable for the case where the considered submodular function is integer-valued, while the function $f_v(\cdot)$ in our case is generally non-integral. Therefore, we adapt the algorithms proposed in~\cite{krause2008robust} to our case and prove new performance bounds, as described in the following.

Our algorithm is based on a ``potential function'' $L$ defined as follows:
\begin{eqnarray}
\forall q\in (0,1], \forall C\subseteq V: L(q,C)=\sum\nolimits_{v\in V}\min\{q,f_v(C)\}
\end{eqnarray}
As $f_v()$ is a submodular function, it can be verified that $L(q,\cdot)$ is also a submodular function for any $q\in (0,1]$. Moreover, the function $L$ has a remarkable property that it can be used to find an upper bound of $\OPT_k^c$, as clarified by the following lemma:
\begin{lemma}
Let $C=\mathsf{Greedy}(\lceil\ln\frac{n}{\epsilon_1 } \rceil k, L(q,\cdot))$. If $L(q,C)< nq-\epsilon_1q$, then $q$ must be an upper bound of $\OPT_k^c$.
\label{lma:judgegreedy}
\end{lemma}

Note that $L(q,C)\geq nq-\epsilon_1 q$ implies that $\min_{v\in V} f_v(C)\geq (1-\epsilon_1) q$. So Lemma~\ref{lma:judgegreedy} actually tells that, if $q<\OPT_k^c$, then we can use function $L$ and the $\mathsf{Greedy}$ procedure to find a clustering $B$ with at most $\lceil\ln\frac{n}{\epsilon_1 } \rceil k$ centering nodes such that $\wht{KC}(B)\geq (1-\epsilon_1)q$. Conversely, if such a clustering cannot be obtained, then we must have $\OPT_k^c\leq q$.

With Lemma~\ref{lma:judgegreedy}, we can use a binary searching process to find an approximate solution to [KCT1], as shown by the $\mathsf{SearchKC}$ algorithm. In the $\mathsf{SearchKC}$ algorithm, we maintain a searching interval $[q_1,q_2]$ (initialized to $[0,1]$), and use Lemma~\ref{lma:judgegreedy} to judge whether $q=\frac{q_1+q_2}{2}$ is an upper bound of $\OPT_k^c$. If we find that $\OPT_k^c\leq q$, then we halves $[q_1,q_2]$ by setting $q_2=q$. Otherwise, we also halves $[q_1,q_2]$ by setting $q_1=q$. As such, we always have
\begin{eqnarray}
\mathrm{OPT}_k^c \leq q_2;~|C^*|\leq \lceil\ln\frac{n}{\epsilon_1 } \rceil k;~\min_{v\in V} f_v(C)\geq (1-\epsilon_1) q_1
\end{eqnarray}
throughout the binary searching process, where the last inequality is due to $L(q_1,C)\geq nq-\epsilon_1 q$. Note that the binary searching process stops when $q_1\geq (1-\epsilon_2)q_2$. So we immediately get the following theorem:

\begin{theorem}
For any $\epsilon, \epsilon_1,\epsilon_2\in (0,1)$ satisfying $1-\epsilon=(1-\epsilon_1)(1-\epsilon_2)$, the $\mathsf{SearchKC}$ algorithm can find a solution that achieves $(1-\epsilon)\OPT_k^c$, using at most $ \left\lceil \ln\frac{n}{\epsilon_1}\right\rceil k$ centers. This algorithm has no more than $\lceil\log_2 \frac{1}{\epsilon_2\cdot \mathrm{OPT}_k^c}\rceil$ iterations. %The time complexity of $\mathsf{SearchKC}$ is at most $\mathcal{O}({k n^2}\log\frac{n}{\epsilon_1}\log \frac{1}{\epsilon_2\cdot {\mathrm{OPT}}_k^c})$.
\end{theorem}

\subsection{$k$-Center Algorithms without an Oracle}

\subsubsection{Approximation Algorithm using Sampling}

In this section, we study whether the algorithm suggested in Sec.~\ref{sec:asimplealgorithm} can be implemented without a connection oracle. Define
\begin{eqnarray}
\wht{d}({\R},u,v)= -\ln \wht{\Pr}[\R,u\sim v];~\wht{d}(\R,v,C)=\min_{u\in C}\wht{d}({\R},u,v)\nonumber
\end{eqnarray}
With this definition, we propose an approximation algorithm as follows:

\begin{algorithm} [ht!]
%   \KwIn{$\epsilon, \epsilon_1,\epsilon_2$; $g(\cdot)$ is a submodular function}%where $0<\sigma+\gamma<\alpha$}
%   \KwOut{A set $S\subseteq V$ and an interval $[a,b]\subseteq [1,n]$}
		Select an arbitrary node $v\in V$ and add it into $C^*$\\
		\While{$|C^*|< k$}
		{
            Find $v^*\in V\backslash C^*$ such that $\wht{d}(\R,v^*,C^*)$ is maximized;\\
            $C\leftarrow C\cup \{v^*\}$\\
            %$q \leftarrow (q_1 + q_2)/2$
		}\label{ln:repeatends}
    $B^*\gets D(\R,C^*)$\\
  \Return{$C^*,B^*$}
    \caption{$\mathsf{SearchKC\_1}(G, k, \R)$}
    \label{alg:searchkc1}
  \end{algorithm}

Next, we study the problem of how to determine the cardinality of $\R$ such that the $k$-clustering $B^*$ returned by the $\mathsf{SearchKC\_1}$ algorithm can achieve a good approximation ratio. We give the following theorem:

\begin{theorem}
Given any $\epsilon,\epsilon_1,\epsilon_2,\delta\in (0,1)$ satisfying $\epsilon=\epsilon_1+\epsilon_2$, and given any set $\R$ of random samples of $G$ satisfying
\begin{eqnarray}
|\R|\geq \max\left\{ \frac{2(1+\epsilon_1)}{3\epsilon_1^2(\OPT_k^c)^2}\ln\frac{n(n-1)}{\delta}, \frac{2(1-\epsilon_1)}{3\epsilon_2^2(\OPT_k^c)^2}\ln\frac{n(n-1)}{\delta} \right\},\nonumber
\end{eqnarray}
the $\mathsf{SearchKC\_1}(G, k, \R)$ can return a $k$-clustering $B^*$ satisfying $\KC(B^*)\geq (1-\epsilon)(\OPT_k^c)^2$ with probability of at least $1-\delta$.
\label{thm:samplenum1}
\end{theorem}

As $\OPT_k^c$ is unknown, we present an algorithm that iteratively ``guesses'' $\OPT_k^c$ until a good solution is found, as shown by Algorithm~\ref{alg:samplingkc1}.

\begin{algorithm} [t]
   \KwIn{$\epsilon=\epsilon_1+\epsilon_2$}%where $0<\sigma+\gamma<\alpha$}
%   \KwOut{A set $S\subseteq V$ and an interval $[a,b]\subseteq [1,n]$}
		$\mathcal{R}\leftarrow \emptyset;\mathcal{R}'\leftarrow \emptyset;i\gets 0$ \label{ln:setab}\\
		\Repeat{$z(u,v)\geq q$}
		{
            $i\gets i+1;~q\leftarrow 2^{-i};~\delta\gets \frac{6\delta}{\pi^2 i^2};~\R\gets \R\cup \R'$\\
			$\ell\gets \max\left\{ \frac{2(1+\epsilon_1)}{3\epsilon_1^2 q^2}\ln\frac{2n(n-1)}{\delta},\frac{2(1-\epsilon_1)}{3\epsilon_2^2 q^2}\ln\frac{2n(n-1)}{\delta}\right\}$;\\
%            $\ell\gets \max\left\{ \ell,\frac{2(1-\epsilon_1)}{3\epsilon_2^2(\OPT_k^c)^2}\ln\frac{n(n-1)}{\delta} \right\};\\
            \If{$|\mathcal{R}|<\ell$}{
                    Add more random samples into $\mathcal{R}$ until $|\mathcal{R}|\geq \ell$\\
                    %$I\leftarrow \mathsf{Initialize}(\mathcal{R})$\\
            }
            $(C^*,B^*)\gets \mathsf{SearchKC\_1}(G, k, \R)$\\
            Generate another set $\R'$ of random samples such that $|\R'|=|\R|$\\
            \ForEach{$(u,v)\in B^*\wedge u\neq v$}{
                $a\gets \ln\frac{2(n-k)}{\delta};~\theta\gets |\R'|$\\
                $z(u,v)\gets \left(\sqrt{\wht{\Pr}[\mathcal{R}', u\sim v]} - \sqrt{\frac{a}{6\theta}}\right)^{2} - \frac{a}{6\theta}$
            }
            $(u^*,v^*)\gets \arg\min_{(u,v)\in B^*\wedge u\neq v} z(u,v)$
%			$q'\leftarrow (1-\epsilon_3)q;~~(u,W)\leftarrow \mathsf{GetFirstNode\_1}(q',\mathcal{R})$\\
%            $C\leftarrow \{u\};~~\eta\leftarrow \mathit{UB}(u);$\\%~~g(\cdot)\gets \wht{L}_{\R}(q,\cdot);$\\
%            \While{$\wht{L}_{\mathcal{R}}(q',C)<nq'-\epsilon_1 q'\wedge |C|<\lceil\ln\frac{n}{\epsilon_1 } \rceil k$}
%                {
%                    $( u,W)\leftarrow \mathsf{GetNextNode}(W,\wht{L}_{\mathcal{R}}(q',\cdot),C)$;\\
%                    $C\leftarrow C\cup \{u\}$
%                }
%            %Find $C$ such that $L(q,C)\geq nq-\epsilon_1$ and $|C|$ is (approximately) minimized.\label{ln:callgrecoverinbisearch}\\
%			\eIf{$|C| \leq \lceil\ln\frac{n}{\epsilon_1 } \rceil k\wedge \wht{L}_{\mathcal{R}}(q',C)\geq nq'-\epsilon_1 q'$ \label{ln:judgeca}}
%			{
%				$C^*\leftarrow C;~~ q_1 \leftarrow q;~~B^*\gets D(\R,C^*)$ \label{ln:setS}
%			}
%			{
%				$q_2 \leftarrow q$ \label{ln:setb}
%			}
            %$q \leftarrow (q_1 + q_2)/2$
		}\label{ln:repeatends}
  \Return{$(C^*,B^*)$}
    \caption{$\mathsf{SamplingKC\_1}(\epsilon_1,\epsilon_2,\epsilon,\delta)$}
    \label{alg:samplingkc1}
  \end{algorithm}

\subsubsection{Sampling for Bi-Criteria Approximation}
A straightforward idea is that, we first generate a set $\R$ of random samples, and then call the $\mathsf{SearchKC}$ algorithm by replacing the function $f_v(\cdot)$ by $\wht{f}_v(\R,\cdot)$. After the $\mathsf{SearchKC}$ algorithm returns a set $C^*$ of centering nodes, we use $B^*=D(\R,C^*)$ as an approximate solution to the $k$-center problem. Clearly,   if $|\R|$ is sufficiently large, then $B^*$ should achieve an approximation ratio close to that we can get with an connectivity oracle. The key problem in this approach, however, is how to determine the cardinality of $\R$. In the following theorem, we propose an upper bound for the number of random samples needed to be generated:

\begin{theorem}
Let $\epsilon_1,\epsilon_2,\epsilon_3,\epsilon$ and $\delta$ be any numbers in $(0,1)$ that satisfy $(1-\epsilon)(1+\epsilon_3)=(1-\epsilon_1)(1-\epsilon_2)(1-\epsilon_3)$.
Let $\mathcal{R}$ be any set of random samples of $G$ such that $$|\mathcal{R}|\geq \frac{2(1+\epsilon_3)}{3\epsilon_3^2 (1-\epsilon) \mathrm{OPT}_k^c}\ln\frac{n^2+n-2k}{2\delta}.$$
Then, we can use the $\mathsf{SearchKC}$ algorithm to find a $k$-clustering $B^*$ with no more than $\lceil \ln \frac{n}{\epsilon_1}\rceil k$ centering nodes such that $\mathit{KC}(B^*)\geq (1-\epsilon)\mathrm{OPT}_{k}^c$ with probability of at least $1-\delta$.
\label{thm:samplingubforkcenter}
\end{theorem}

As $\OPT_k^c$ is unknown, we need to find a lower bound of $\OPT_k^c$ to determine the cardinality of $\R$. Recall that we have used a trivial lower bound $k/n$ for $\OPT_k^m$ in the $k$-median problem. However, it is hard to find an ideal lower bound for $\OPT_k^c$. A trivial lower bound for $\OPT_k^c$ is the production of the existence probabilities of all the edges in $E$, but this lower bound could be too small and hence results in a large number of generated random samples. In the sequel we will provides more efficient algorithm for $k$-Center.

%\subsubsection{A More Efficient Sampling Algorithm}

\begin{algorithm} [t]
   \KwIn{$(1-\epsilon)(1+\epsilon_3)=(1-\epsilon_1)(1-\epsilon_2)(1-\epsilon_3)$}%where $0<\sigma+\gamma<\alpha$}
%   \KwOut{A set $S\subseteq V$ and an interval $[a,b]\subseteq [1,n]$}
		$[q_1,q_2]\leftarrow [0,1]; C^*\leftarrow \emptyset; \mathcal{R}\leftarrow \emptyset;i\gets 0$ \label{ln:setab}\\
		\Repeat{$q_1\geq (1-\epsilon_2)q_2$}
		{
            $i\gets i+1;~q\leftarrow \frac{q_1+q_2}{2};~\delta\gets \frac{6\delta}{\pi^2 i^2}$\\
			$\ell\leftarrow \lceil\frac{2(1+\epsilon_3)}{3\epsilon_3^2 (1-\epsilon)q}\ln\frac{n^2+n-2k}{2\delta}\rceil$;\\
            \If{$|\mathcal{R}|<\ell$}{
                    Add more random samples into $\mathcal{R}$ until $|\mathcal{R}|=\ell$\\
                    %$I\leftarrow \mathsf{Initialize}(\mathcal{R})$\\
            }
			$q'\leftarrow (1-\epsilon_3)q;~~(u,W)\leftarrow \mathsf{GetFirstNode\_1}(q',\mathcal{R})$\\
            $C\leftarrow \{u\};~~\eta\leftarrow \mathit{UB}(u);$\\%~~g(\cdot)\gets \wht{L}_{\R}(q,\cdot);$\\
            \While{$\wht{L}_{\mathcal{R}}(q',C)<nq'-\epsilon_1 q'\wedge |C|<\lceil\ln\frac{n}{\epsilon_1 } \rceil k$}
                {
                    $( u,W)\leftarrow \mathsf{GetNextNode}(W,\wht{L}_{\mathcal{R}}(q',\cdot),C)$;\\
                    $C\leftarrow C\cup \{u\}$
                }
            %Find $C$ such that $L(q,C)\geq nq-\epsilon_1$ and $|C|$ is (approximately) minimized.\label{ln:callgrecoverinbisearch}\\
			\eIf{$|C| \leq \lceil\ln\frac{n}{\epsilon_1 } \rceil k\wedge \wht{L}_{\mathcal{R}}(q',C)\geq nq'-\epsilon_1 q'$ \label{ln:judgeca}}
			{
				$C^*\leftarrow C;~~ q_1 \leftarrow q;~~B^*\gets D(\R,C^*)$ \label{ln:setS}
			}
			{
				$q_2 \leftarrow q$ \label{ln:setb}
			}
            %$q \leftarrow (q_1 + q_2)/2$
		}\label{ln:repeatends}
  \Return{$(C^*,B^*)$}
    \caption{$\mathsf{SearchKC+}(\epsilon_1,\epsilon_2,\epsilon_3,\epsilon,\delta)$}
    \label{alg:samplingbisearch}
  \end{algorithm}

%\begin{algorithm} [t]
%%   \KwIn{$\epsilon, \epsilon_1,\epsilon_2$; $g(\cdot)$ is a submodular function}%where $0<\sigma+\gamma<\alpha$}
%%   \KwOut{A set $S\subseteq V$ and an interval $[a,b]\subseteq [1,n]$}
%        \ForEach{$v\in V$}{
%            $\mathit{UB}(v)\leftarrow$ the summation of the sizes of the connected components in $\mathcal{R}$ that contain $v$;
%		}
%        Sort $V$ into $u_1,\cdots,u_n$ such that $\mathit{UB}(u_1)\geq \cdots\geq \mathit{UB}(u_n)$;\\
%  \Return{$I=\langle u_1,\cdots,u_n\rangle$}
%    \caption{$\mathsf{Initialize}(\mathcal{R})$}
%    \label{alg:bisearch}
%  \end{algorithm}

\begin{algorithm} [ht!]
%   \KwIn{$\epsilon, \epsilon_1,\epsilon_2$; $g(\cdot)$ is a submodular function}%where $0<\sigma+\gamma<\alpha$}
%   \KwOut{A set $S\subseteq V$ and an interval $[a,b]\subseteq [1,n]$}
    Compute the node list $W$ and the values of $\forall v\in V: \UB(v)$ using Lines 1-3 of $\mathsf{GetFirstNode}(\mathcal{R})$\\
    $(u,W)\leftarrow \mathsf{GetNextNode}(W,\wht{L}_{\mathcal{R}}(q,\cdot),\emptyset)$\\
    \Return{$(u,W)$}
    \caption{$\mathsf{GetFirstNode\_1}(q,\mathcal{R})$}
    \label{alg:bisearch}
\end{algorithm}

We first study whether we can apply the OPIM framework~\cite{Tang2018} to accelerate our algorithm.
For any $B\in \mathcal{S}_k^G$, any $C\subseteq V$ and any $v\in V$, we define
\begin{eqnarray}
&&{\wht{\mathit{KC}}}({\mathcal{R}},B)={\min\nolimits_{(u,v)\in B}\wht{\mathrm{Pr}}[{\mathcal{R}}, u\sim v]}\nonumber
%&&\wht{f}_v(\mathcal{R},C)=\max\{\wht{\mathrm{Pr}}}[{\mathcal{R},u\sim v]\mid u\in C\}
\end{eqnarray}
Let $(u^o,v^o)$ denote a cluster link in $B^o$ such that $\wht{\Pr}[\R, u^o\sim v^o]= {\wht{\mathit{KC}}}({\mathcal{R}},B^o)$.
The OPIM framework requires that we can find an upper bound of $\wht{\Pr}[\R, u^o\sim v^o]$ using $B^*$, under the purpose that we can get an upper bound of $\OPT_k^c$. This idea, however, cannot be applied to the $k$-center problem. To explain, note that $\wht{\mathrm{Pr}}[\mathcal{R},u^o\sim v^o]$ could be larger than $\wht{KC}(\mathcal{R},{B}^o)$, while we can only guarantee that $\wht{KC}(\mathcal{R},{B}^o)$ is no more than $\wht{\mathrm{Pr}}[\mathcal{R},u\sim v]$ for all $(u,v)\in B^*$. Therefore, it is possible that $\wht{\mathrm{Pr}}[\mathcal{R},u^o\sim v^o]$ is larger than the estimated probability of any cluster link in $B^*$. Therefore, the OPIM framework cannot be applied to the $k$-center problem.

Based on the above observation, we propose a method to judge whether $q\geq \OPT_k^c$ using a relatively small number of random samples. For any $q\in (0,1]$ and any $C\subseteq V$, we define
\begin{eqnarray}
\wht{L}_{\mathcal{R}}(q,C)=\sum\nolimits_{v\in V}\min\{q,\wht{f}_v(\mathcal{R},C)\}
%\wht{z}(d,A)= n|\mathcal{F}| |l(A)\cap Q| +\sum_{y\in \mathcal{F}}\min\{d,\wht{y}(\mathcal{R}_y,A)\}
\end{eqnarray}
and we prove the following lemma:
%and define a ``dual problem'' for [KC1] as
%    \begin{eqnarray}
%        \mathbf{Minimize}&&~~|C| \qquad\qquad\qquad\qquad \mathbf{[kCenter\_Coverage]}\nonumber\\
%        \mathbf{s.t.}&&~~\wht{L}_{\mathcal{R}}(q,C)\geq nq \label{eqn:robustIEOconstraint}
%        %&&~~h(A)\geq n|Q|+\xi \nonumber
%    \end{eqnarray}
%\label{def:defofz}

%\begin{lemma}
%If $\mathrm{OPT}_k^c\geq q$, then the greedy procedure in Lines~\ref{}-\ref{} of the $\mathsf{SearchKC}$ algorithm can find $C\subseteq V$ such that $|C|\leq \lceil\ln\frac{n}{\epsilon_1 } \rceil k$ and $\wht{L}_{\R}(q,C)\geq (n-\epsilon_1)(1-\epsilon_3)q$ and $\min_{v\in V}\wht{f}_v(\R,C)\geq (1-\epsilon_1)(1-\epsilon_3)q$ with probability of at least $1-\delta$.
%\end{lemma}

\begin{lemma}
Let $q, \delta$ be any numbers in $(0,1)$ and $\R$ be any set of weakly dependant random samples of $G$. If $\OPT_k^c>q$ and  $|\mathcal{R}|\geq \frac{2(1+\epsilon_3)}{3\epsilon_3^2 q}\ln\frac{n-k}{\delta}$, then we must have $$\Pr[\wht{L}_{\R}(q,C)\geq (n-\epsilon_1)(1-\epsilon_3)q]\geq 1-\delta,$$ where $C=\mathsf{Greedy}(\lceil\ln\frac{n}{\epsilon_1 } \rceil k, \wht{L}_{\R}(q,\cdot))$.
\label{lma:judgegreedyinrandomworld1}
\end{lemma}

%\begin{lemma}
%For any $q$ tested in the algorithm, we have $\Pr[|C^*_q|>\lceil\ln\frac{n}{\epsilon_1 } \rceil k\wedge \mathrm{OPT}_k^c\geq q]\leq \delta_q$, where $\delta_q$ denotes the value of $\delta$ when $q$ is tested.
%\end{lemma}
%\begin{proof}
%Let $\mathcal{R}_q$ denote the set $\mathcal{R}$ of random samples when $q$ is tested. Let $B^o_q$ denote a $k$-clustering in $G$ such that $\wht{\KC}(\R_q, B^o_q)$ is maximized. When $|C^*_q|> \lceil\ln\frac{n}{\epsilon_1 } \rceil k$, we must have $\wht{\KC}(\mathcal{R}_q,B^o_q)< (1-\epsilon_3)q$. Besides, we have $\wht{\KC}(\mathcal{R}_q,B^o)\leq \wht{\KC}(\mathcal{R}_q,B^o_q)$. So we can get
%\begin{eqnarray}
%&&\Pr[|C^*_q|>\lceil\ln\frac{n}{\epsilon_1 } \rceil k\wedge \mathrm{OPT}_k^c\geq q]\nonumber\\
%&\leq& \Pr[\wht{\KC}(\mathcal{R}_q,B^o)< (1-\epsilon_3)\mathrm{OPT}_k^c \wedge \mathrm{OPT}_k^c\geq q]\nonumber\\
%&\leq& \exp\left(-\frac{3\epsilon_3^2q}{2(1+\epsilon_3)}|\mathcal{R}_q|\right)\leq \delta_q
%\end{eqnarray}
%\end{proof}

Note that we must have $\wht{KC}(\R,B^o)\geq (1-\varepsilon_3)\OPT_k^c$ with probability of at least $1-\delta$ when $|\mathcal{R}|\geq \frac{2(1+\epsilon_3)}{3\epsilon_3^2 q}\ln\frac{n-k}{\delta}$. So the proof of Lemma~\ref{lma:judgegreedyinrandomworld1} is similar to that of Lemma~\ref{lma:judgegreedy}.

%Lemma~\ref{lma:judgegreedyinrandomworld1} implies that, if $\wht{L}_{\R}(q,C)< (n-\epsilon_1)(1-\epsilon_3)q$ and $|\mathcal{R}|\geq \frac{2(1+\epsilon_3)}{3\epsilon_3^2 q}\ln\frac{n-k}{\delta}$, then we can judge that $\OPT_k^c>q$, and the probability that such a judgement is wrong is no more than $\delta$.

With Lemma~\ref{lma:judgegreedyinrandomworld1}, we propose a binary searching process similar to that in Algorithm~\ref{alg:bisearch} to find a bi-criteria approximation solution to the $k$-center problem, as shown by the $\mathsf{SearchKC+}$ algorithm.

Similar to the $\mathsf{SearchKC}$ algorithm, the $\mathsf{SearchKC+}$ algorithm also maintains a searching interval $[q_1,q_2]$ and halves this interval in each iteration. The main difference between $\mathsf{SearchKC}$ and $\mathsf{SearchKC+}$ is that we have replaced the function ${L}(q,\cdot)$ by $\wht{L}_{\R}(q,\cdot)$ and used Lemma~\ref{lma:judgegreedyinrandomworld1} to guide the direction of the binary searching process. More specifically, in each iteration $i$, we set $q=\frac{q_1+q_2}{2}$, and generate a set $\R$ of $\frac{2(1+\epsilon_3)}{3\epsilon_3^2 q}\ln\frac{n-k}{\delta}$ random samples. Then we greedily select at most $\ln$ nodes into $C$. If $\wht{L}_{\R}(q,C)\geq (n-\epsilon_1)(1-\epsilon_3)q$, then we can judge that $\OPT_k^c>q$, and the probability that such a judgement is wrong is no more than $\frac{6\delta}{\pi^2 i^2}$ due to Lemma~\ref{lma:concentrationbound}. By the union bound, the probability that we have searched the wrong direction is no more than $\sum_{i=1}^\infty \frac{6\delta}{\pi^2 i^2} =\delta$. If we never search the wrong direction, then we use similar reasoning with that in Sec.~\ref{sec:bisearchwithoracle} to know that we have got a good approximation solution. More specifically, we can prove:

\begin{theorem}
For any $\epsilon, \epsilon_1,\epsilon_2, \epsilon_3\in (0,1)$ satisfying $1-\epsilon=(1-\epsilon_1)(1-\epsilon_2)(1-\epsilon_3)$, the $\mathsf{SearchKC+}$ algorithm can find a solution that achieves $(1-\epsilon){\mathrm{OPT}}_k^c$ with probability of at least $1-\delta$, using at most $\left\lceil \ln\frac{n}{\epsilon_1}\right\rceil k$ centers. This algorithm has no more than $\lceil \log_2 \frac{1}{\epsilon_2\cdot {\mathrm{OPT}}_k^c} \rceil$ iterations.
\end{theorem}

Similar to the $\mathsf{SearchKM+}$ algorithm, $\mathsf{SearchKC+}$ also leverages the CELF framework to reduce the number of evaluating $\wht{L}_{\R}(q,\cdot)$. However, it uses a different procedure (i.e., the $\mathsf{GetFirstNode\_1}$ algorithm) to address the ``cold start'' problem.  %in the CELF framework. %More specifically, we first compute $f_v(\R,\{u\})$ for each $u\in V$ in

Finally, we ask the question whether the $\mathsf{SearchKC+}$ could generate ``too many'' random samples, compared with the upper bound proposed in Theorem~\ref{thm:samplingubforkcenter}. To answer this question, we prove the following theorem:

%we greedily selected  we can judge whether the If $\mathcal{E}_q$ does not happen for every tested $q$ in the algorithm, then we must have
%\begin{eqnarray}
%&&\mathrm{OPT}_k^c \leq q_2\leq q_1+\epsilon_2 q_2\nonumber\\
%&&\wht{\KC}(\R_{q_1},B^*)\geq (1-\epsilon_1)(1-\epsilon_3)q_1
%\end{eqnarray}
%when the algorithm terminates, which gives us $\wht{\KC}(\R_{q_1},B^*)\geq (1-\epsilon_1)(1-\epsilon_2)(1-\epsilon_3)\mathrm{OPT}_k^c$

\begin{theorem}
The expected number of random samples generated in the $\mathsf{SearchKC+}$ algorithm is at most $\mathcal{O}\left(\frac{1}{\epsilon^2 \OPT_k^c}\left(\ln \frac{n}{\delta}+\ln\ln \frac{1}{\epsilon\OPT_k^c}\right)\right)$
\label{thm:expectednumberofsamples}
\end{theorem}

Note that the bound shown in Theorem~\ref{thm:expectednumberofsamples} is very close to that shown in Theorem~\ref{thm:samplingubforkcenter}, and it has only introduced an additional $\ln\ln \frac{1}{\epsilon\OPT_k^c}$ factor. This demonstrates that the $\mathsf{SearchKC+}$ algorithm would not generate a lot of unnecessary random samples.

%we need to generate $T_{max}=\mathcal{O}\left(\frac{1}{\epsilon^2 \OPT_k^c}\ln \frac{n}{\delta}\right)$ random samples in the basic sampling algorithm introduced in Sec.~\ref{}, but under the assumption that $\OPT_k^c$ is known. Therefore, Lemma~\ref{} reveals that the $\mathsf{SearchKC+}$ algorithm would not generate ``too many'' random samples when $\OPT_k^c$ is unknown, as the bound introduced in Lemma~\ref{} only introduces an additional $\ln\ln \frac{1}{\epsilon\OPT_k^c}$ factor compared to $T_{max}$.
%
%
%
%Recall that we need to generate $T_{max}=\mathcal{O}\left(\frac{1}{\epsilon^2 \OPT_k^c}\ln \frac{n}{\delta}\right)$ random samples in the basic sampling algorithm introduced in Sec.~\ref{}, but under the assumption that $\OPT_k^c$ is known. Therefore, Lemma~\ref{} reveals that the $\mathsf{SearchKC+}$ algorithm would not generate ``too many'' random samples when $\OPT_k^c$ is unknown, as the bound introduced in Lemma~\ref{} only introduces an additional $\ln\ln \frac{1}{\epsilon\OPT_k^c}$ factor compared to $T_{max}$.

\section{Conclusion}

We have studied the $k$-median and $k$-center problems in uncertain graphs. We have analyzed the complexity of these problems and proposed efficient algorithms with improved approximation ratios compared with the prior art.

\begin{appendix}
\section{Missing Lemmas} \label{sec:append}

\begin{lemma}
Given any set $A\subseteq V\times V$, any set $\mathcal{R}$ of weakly-dependant random samples of $G$ and any positive number $\varepsilon$, we have
    \begin{eqnarray}
    &&\mathrm{Pr}\left[\wht{\Upsilon}(A)-\Upsilon(A)\geq\varepsilon\right]\leq\exp\left\{-\frac{3\varepsilon^2|\mathcal{R}|}{2(\varepsilon+\Upsilon(A))}\right\} \nonumber
    \end{eqnarray}
where $\Upsilon(A)={\sum_{(u,v)\in A}{\Pr}[u\sim v]}/{|A|}$ and $\wht{\Upsilon}(A)=\sum_{(u,v)\in A}$ $\wht{\Pr}[\mathcal{R},u\sim v]/{|A|}$
\label{lma:concentrationbound}
\end{lemma}

\section{Missing Proofs}

\subsection{Proof of Theorem~\ref{thm:submodularitythm}}
\begin{proof}
For any $X\subseteq Y\subseteq V$ and any $x\in V\backslash Y$, we have
\begin{eqnarray}
&&f_v(X)=\max\{\mathrm{Pr}[u\sim v]\mid u\in X\}\nonumber\\
&\leq& \max\{\mathrm{Pr}[u\sim v]\mid u\in Y\}=f_v(Y)
\end{eqnarray}
So $f_v(\cdot)$ is monotone. Note that $$f_v(X\cup \{x\})=\max\{\mathrm{Pr}[x\sim v], f_v(X)\}$$ So we have:

1) if $\mathrm{Pr}[x\sim v]\geq f_v(Y)=\max\{\mathrm{Pr}[u\sim v]\mid u\in Y\}$, then we also have $\mathrm{Pr}[x\sim v]\geq f_v(X)$, so
\begin{eqnarray}
&&f_v(X\cup \{x\})-f_v(X)\nonumber\\
&=& \mathrm{Pr}[x\sim v]-f_v(X) \nonumber\\
&\geq& \mathrm{Pr}[x\sim v]-f_v(Y) \nonumber\\
&=& f_v(Y\cup \{x\})-f_v(Y)
\end{eqnarray}

2) if $\mathrm{Pr}[x\sim v]< f_v(Y)=\max\{\mathrm{Pr}[u\sim v]\mid u\in Y\}$, then we have $f_v(Y\cup \{x\})=f_v(Y)$, so we also have
\begin{eqnarray}
f_v(X\cup \{x\})-f_v(X)\geq 0= f_v(Y\cup \{x\})-f_v(Y)
\end{eqnarray}
\end{proof}

\subsection{Proof of Theorem~\ref{thm:samplingubforkmedian}}
\begin{proof}
%For any set $C$ of $k$ nodes in $V$ and any set $\mathcal{R}$ of random samples of $G$, we can generate a $k$-cluster by regarding the nodes in $C$ as the centers and assigning each node $v\in V\backslash C$ a center $u\in C$ such that 1) $u=\arg\max_{u\in C}\wht{\mathrm{Pr}}[\mathcal{R},u\sim v]$ and 2) the node-id of $u$ is minimized. For convenience, we use  $D(\mathcal{R},C)$ to denote the signature of such a $k$-cluster generated by $\mathcal{R}$ and $C$.

Let $\mathcal{S}_{\mathcal{R}}=\{{D}(\mathcal{R},C)||C|=k\wedge C\subseteq V\}$. Then we must have $|\mathcal{S}_\mathcal{R}|\leq {n\choose k}$. Let $\wht{A}^o$ denote the signature in $\mathcal{S}_{\mathcal{R}}$ such that $\wht{KM}(\mathcal{R},\wht{A}^o)$ is maximized. Let  $A^o$ denote the signature of an optimal clustering of the k-median problem. Note that $A^o$ is not necessarily in $\mathcal{S}_{\mathcal{R}}$, but there must exist certain $A'\in \mathcal{S}_{\mathcal{R}}$ such that $\wht{\mathit{KM}}(\mathcal{R},A')\geq \wht{\mathit{KM}}(\mathcal{R},A^o)$. Therefore, we must have:
\begin{eqnarray}
\wht{KM}(\mathcal{R},A^*)\geq (1-1/e)\wht{KM}(\mathcal{R},\wht{A}^o)\geq (1-1/e)\wht{\mathit{KM}}(\mathcal{R},A^o) \nonumber
\end{eqnarray}

With the above equation, we can ensure that $A^*$ has a $1-1/e-\epsilon$ approximation ratio by adopting the following sampling method: we generate sufficient random samples in $\mathcal{R}$ such that $\wht{\mathit{KM}}(\mathcal{R},A)$ is an estimation of $\mathit{KM}(A)$ within an absolute error of $\beta=\frac{e}{2e-1}\cdot \epsilon\mathrm{OPT}_{km}$ for any $A\in \mathcal{S}_{\mathcal{R}}\cup\{A^o\}$. More specifically, when $|\mathcal{R}|\geq \frac{2(2e-1)(e\epsilon+2e-1)}{3e^2\epsilon^2\mathrm{OPT}_{km}}\ln\frac{{n\choose k}+1}{\delta}$, we must have
\begin{eqnarray}
&&\mathrm{Pr} \left[ \exists A\in \mathcal{S}_{\mathcal{R}}: \mathit{KM}(A)< \wht{KM}(\mathcal{R},A)-\beta\right] \nonumber\\
&\leq& \sum_{A\in \mathcal{A}_{\mathcal{R}}} \mathrm{Pr}[\mathit{KM}(A)< \wht{KM}(\mathcal{R},A)-\beta]\nonumber\\
&\leq& \frac{{n\choose k}\delta}{{n\choose k}+1}
\end{eqnarray}
and
\begin{eqnarray}
\mathrm{Pr} \left[ \wht{\mathit{KM}}(\mathcal{R},A^o)<\mathit{KM}(A^o)-\beta \right]\leq \frac{\delta}{{n\choose k}+1}
\end{eqnarray}
Moreover, as $A^*\in \mathcal{A}_{\mathcal{R}}$, we can use the union bound to get that, with probability of at least $1-\delta$, we have:
\begin{eqnarray}
\mathit{KM}(A^*)&\geq& \wht{KM}(\mathcal{R},A^*)-\beta\nonumber\\
&\geq& (1-1/e)\wht{\mathit{KM}}(\mathcal{R},A^o)-\beta\nonumber\\
&\geq& (1-1/e)\left({{KM}}(A^o)-\beta\right)-\beta\nonumber\\
&=& (1-1/e-\epsilon)\mathrm{OPT}_{k}^m\nonumber
\end{eqnarray}
Hence, the lemma follows.
\end{proof}

\subsection{Proof of Theorem~\ref{thm:inapprofkcenter}}
\begin{proof}
Again, we prove the theorem by a reduction from the NP-hard dominating set problem. Given an instance $G=(V,E)$ of the dominating set problem where $|V|=n$, we can construct an uncertain graph by setting $p(e)=q=\frac{\alpha}{2(1+\alpha)n}$ for all $e\in E$. Suppose that there exists a polynomial-time algorithm that can find $B\in \mathcal{S}_k^G$ such that $\mathit{KC}(B)\geq \alpha \mathit{KC}(\tilde{B})$, where $\tilde{B}$ is an optimal solution to the $k$-center problem in the uncertain $G$ constructed above. We will prove that: there exists a dominating set $S$ in $G$ with $|S|=k$ if and only if $\mathit{KC}(B)\geq \alpha q$. Indeed, if such a dominating set $S$ exists, then we must have $\mathit{KC}(\tilde{B})\geq q$ and hence $\mathit{KC}(B)\geq \alpha q$. Conversely, if there does not exist such a dominating set $S$, then there must exist a cluster link $(u,v)$ in $B$ such that $u\neq v$ and $u$ is not adjacent to $v$. Thus, we can use similar reasoning as that in Theorem~\ref{thm:hardnesskmedian} to prove that
\begin{eqnarray}
\mathit{KC}(B)\leq \mathrm{\mathrm{Pr}}[u\sim v]\leq \frac{nq^2}{1-nq}<\alpha q
\end{eqnarray}
The above reasoning implies that, if there exists a polynomial-time algorithm to the $k$-center problem with any approximation ratio $\alpha$, then the dominating set can also be optimally solved in polynomial time. Hence, the theorem follows.
\end{proof}

\subsection{Proof of Theorem~\ref{thm:kcboundedbylnn}}
\begin{proof}
Given any graph $G=(V,E)$ with $|V|=n$, we can construct an uncertain graph by setting $p(e)=q=\frac{1}{3n}$. Under this setting, we can use similar reasoning with that in Theorem~\ref{thm:inapprofkcenter} to prove that: for any clustering $B$ of $G$, the set of center nodes in $B$ is a dominating set of $G$ if and only if $\mathit{KC}(B)\geq q$.

Suppose that the cardinality of the minimum dominating set in $G$ is $k^*$ ($k^*$ is unknown), and suppose by contradiction that there exists a bi-criteria approximation algorithm $\mathcal{A}$ that achieves the properties described by the theorem. Then we can find a dominating set of $G$ as follows. We run $\mathcal{A}$ for all $k\in [n]$. Let $C_k$ denote the set of center nodes returned by $\mathcal{A}$ for any $k\in [n]$. We then return the set $C\in \{C_1,\cdots, C_n\}$ such that $C$ is a dominating set of $G$ and $|C|$ is minimized.

Note that $\mathrm{OPT}_{k^*}^c\geq q$ in such a case. So $C_{k^*}$ must be a dominating set of $G$ according to the above reasoning. Moreover, we have $|C_{k^*}|< k^*\ln n$ and hence $|C|< k^*\ln n$. This implies that we have built a polynomial-time dominating set algorithm with an approximation ratio less than $\ln n$. However, it is proved in~\cite{feige1998threshold} that such a dominating set algorithm should not exist unless P=NP. Therefore, we got a contradiction, which proves the theorem.
\end{proof}

\subsection{Proof of Lemma~\ref{lma:judgegreedy}}
\begin{proof}
Consider the following optimization problem:
    \begin{eqnarray}
        \mathbf{Minimize}&&~~|C| \qquad\qquad\qquad\qquad \mathbf{[KCT\_COVER]}\nonumber\\
        \mathbf{s.t.}&&~~L(q,C)\geq nq;~~C\subseteq V \label{eqn:robustIEOconstraint}
    \end{eqnarray}
Note that $L(q,\cdot)$ is a monotone and submodular function and $L(q,V)=nq$. So this problem is actually a ``submodular set cover'' problem. Suppose that an optimal solution to [KCT\_COVER] is $C^o_{cover}$. We can use a greedy algorithm~\cite{Goyal2013} to find $C'$ such that
\begin{eqnarray}
&&|C'|\leq \left\lceil \ln\frac{n}{\epsilon_1}\right\rceil |{C}^o_{cover}|;~~L(q,C')\geq nq-\epsilon_1 q
\end{eqnarray}
Now suppose by contradiction that $\mathrm{OPT}_k^c> q$, then we must have
\begin{eqnarray}
\min\nolimits_{v\in V} f_v(C^o_{kc})=\mathrm{OPT}_k^c >q,
\end{eqnarray}
where $C^o_{kc}$ denote the set of centering nodes in $A^o$. So we have
\begin{eqnarray}
L(q,C^o_{kc})= \sum\nolimits_{v\in V}\min\{q,f_v(C^o_{kc})\}=qn
\end{eqnarray}
This implies that $C^o_{kc}$ is a feasible solution to [KCT\_COVER]. Therefore, we must have
\begin{eqnarray}
|C'|\leq \left\lceil \ln\frac{n}{\epsilon_1}\right\rceil |{C}^o_{cover}|\leq \left\lceil \ln\frac{n}{\epsilon_1}\right\rceil |{C}^o_{kc}|\leq \left\lceil \ln\frac{n}{\epsilon_1}\right\rceil k \nonumber
\end{eqnarray}
and hence $C'\subseteq C$. This implies $L(q,C)\geq nq-\epsilon_1q$, a contradiction. Hence, the lemma follows.
\end{proof}

\subsection{Proof of Theorem~\ref{thm:samplenum1}}
\begin{proof}
Let $Z=\{(u,v)|u,v\in V\wedge u\neq v\}$. Define the sets $Q_{\R},Q'_{\R}$ and the events $\mathcal{E}_1,\mathcal{E}_2$ as
\begin{eqnarray}
&&Q_{\R}=\{(u,v)|(u,v)\in Z\wedge \Pr[u\sim v]\geq (\OPT_k^c)^2\}\nonumber\\
&&\mathcal{E}_1=\{\forall (u,v)\in Q_{\R}:\wht{\Pr}[\R, u\sim v]\geq (1-\epsilon_1)(\OPT_k^c)^2\}\nonumber\\
&&Q'_{\R}=\{(u,v)|(u,v)\in Z\wedge \Pr[u\sim v]< (1-\epsilon)(\OPT_k^c)^2\}\nonumber\\
&&\mathcal{E}_2=\{\forall (u,v)\in Q'_{\R}:\wht{\Pr}[\R, u\sim v]< (1-\epsilon_1)(\OPT_k^c)^2\}\nonumber
\end{eqnarray}
In the sequel, we will prove that: when $\mathcal{E}_1$ and $\mathcal{E}_2$ both happen, then we must have $\KC(B^*)\geq (1-\epsilon)(\OPT_k^c)^2$.

Suppose that the nodes sequentially selected by $\mathsf{SearchKC\_1}$ are $v_1,v_2,\cdots,v_k$. Let $C^*_i=\{v_1,\cdots,v_i\}$. Let $v_{k+1}$ be a node in $V\backslash C^*_k$ such that $\wht{d}(\R,v_{k+1},C^*_k)$ is maximized. Let $l_i=\wht{d}(\R,v_{i+1},C^*_{i})$. We first prove that $l_k\leq -2\ln \KC(B^o)-\ln(1-\epsilon)$. It can be seen that $l_1,l_2,\cdots, l_{k}$ are non-increasing. Suppose by contradiction that $l_k>-2\ln \KC(B^o)-\ln(1-\epsilon)$, then there must exist two nodes $v_i$ and $v_j$ in one cluster of $B^o$ such that $\wht{d}({\R},v_i,v_j)\geq -2\ln \KC(B^o)-\ln(1-\epsilon)$, which implies that
\begin{eqnarray}
\wht{\Pr}[\R,v_i\sim v_j] <(1-\epsilon) (\OPT_k^c)^2
\end{eqnarray}
However, as $v_i$ and $v_j$ is in the same cluster of $B^o$, we must have $\Pr[v_i\sim v_j]\geq (\OPT_k^c)^2$, which contradicts the assumption that the event $\mathcal{E}_1$ happens.

As $l_k\leq -2\ln \KC(B^o)-\ln(1-\epsilon)$, we must have
\begin{eqnarray}
\forall (u,v)\in B^*: \wht{\Pr}[\R,u\sim v] \geq (1-\epsilon_1) (\OPT_k^c)^2
\end{eqnarray}
As $\mathcal{E}_2$ happens, we must have
\begin{eqnarray}
\KC(B^*) \geq (1-\epsilon) (\OPT_k^c)^2
\end{eqnarray}
Now the problem left is to prove that
\begin{eqnarray}
\Pr[\neg \mathcal{E}_1\vee \neg\mathcal{E}_2]\leq \delta
\end{eqnarray}
This can be proved by using the union bound and Lemma~\ref{lma:concentrationbound}. Hence the theorem follows.
\end{proof}

\subsection{Proof of Theorem~\ref{thm:samplingubforkcenter}}
\begin{proof}
Let $(u^*,v^*)$ be the cluster-link in $B^*$ such that ${\mathrm{Pr}}[u\sim v]$ is minimized. Let $(\wht{u},\wht{v})$ be the cluster-link in $B^*$ such that $\wht{\mathrm{Pr}}[\mathcal{R},\wht{u}\sim \wht{v}]$ is minimized. Let $B^o$ be the signature of an optimal $k$-clustering in $G$ such that $\mathit{KC}(B^o)$ is maximized. Let $(\wht{u}',\wht{v}')$ be the cluster-link in $B^o$ such that $\wht{\mathrm{Pr}}[\mathcal{R},\wht{u}'\sim \wht{v}']$ is minimized.  Let $({u}^o,{v}^o)$ be the cluster-link in $B^o$ such that ${\mathrm{Pr}}[{u}^o\sim {v}^o]$ is minimized.  Let $\mathcal{E}_1$ denote the following event:
\begin{eqnarray}
&&\mathcal{E}_1=\{\wht{\mathrm{Pr}}[\mathcal{R},\wht{u}'\sim \wht{v}']\geq (1-\epsilon_3){\mathrm{Pr}}[\wht{u}'\sim \wht{v}']\}~~\label{eqn:uhatisgood}
%&&\mathcal{E}_2=\{(1+\epsilon_3){\mathrm{Pr}}[u^*\sim v^*]\geq \wht{\mathrm{Pr}}[\mathcal{R},u^*\sim v^*]\}~~\label{eqn:ustarisgood}\\
%&&\mathcal{E}_3=\{{\mathrm{Pr}}[u^*\sim v^*]< (1-\epsilon)\mathrm{OPT}_k^c\}~~\label{eqn:ustarisgood}
\end{eqnarray}
If $\mathcal{E}_1$ hold, then we must have
\begin{eqnarray}
&& \wht{\mathrm{Pr}}[\mathcal{R},u^*\sim v^*]\geq \wht{\mathrm{Pr}}[\mathcal{R},\wht{u}\sim \wht{v}] \nonumber\\
&\geq& (1-\epsilon_1)(1-\epsilon_2)\wht{KC}(\mathcal{R},\wht{B}^o) \label{eqn:duetothekcugalg}\\
&\geq& (1-\epsilon_1)(1-\epsilon_2)\wht{KC}(\mathcal{R},{B}^o)\nonumber\\
&=& (1-\epsilon_1)(1-\epsilon_2)\wht{\mathrm{Pr}}[\mathcal{R},\wht{u}'\sim \wht{v}']\nonumber\\
&\geq& (1-\epsilon_1)(1-\epsilon_2)(1-\epsilon_3){\mathrm{Pr}}[\wht{u}'\sim \wht{v}'] \nonumber\\
&\geq& (1-\epsilon_1)(1-\epsilon_2)(1-\epsilon_3){\mathrm{Pr}}[{u}^o\sim {v}^o]   \nonumber\\
&=& (1-\epsilon)(1+\epsilon_3)\mathrm{OPT}_k^c  \nonumber
\end{eqnarray}
where \eqref{eqn:duetothekcugalg} is due to the performance guarantee of $\mathsf{SearchKC}$.

Note that there are at most $n-k$ possible choices for $(\wht{u}',\wht{v}')$. Therefore, when $|\mathcal{R}|\geq \frac{2(1+\epsilon_3)}{3\epsilon_3^2 \mathrm{OPT}_k^c}\ln\frac{n-k}{\delta_1}$, we can get
\begin{eqnarray}
&&\Pr[\neg \mathcal{E}_1]\leq (n-k)\exp\left(-\frac{3\epsilon_3^2\Pr[\wht{u}'\sim \wht{v}']}{2(1+\epsilon_3)}|\mathcal{R}|\right) \nonumber\\
&\leq& \frac{2(n-k)}{n^2+n-2k}\delta,
\end{eqnarray}
where we have used the the union bound and Lemma~\ref{lma:concentrationbound}.

Let $\mathcal{K}=\{(u,v)|u,v\in V\wedge \Pr[u\sim v]<(1-\epsilon)\mathrm{OPT}_k^c\}$. So we have $|\mathcal{K}|\leq {n\choose 2}$. When $|\mathcal{R}|\geq \frac{2(1+\epsilon_3)}{3\epsilon_3^2 (1-\epsilon) \mathrm{OPT}_k^c}\ln\frac{{n\choose 2}}{\delta_2}$, we can use the the union bound and Lemma~\ref{lma:concentrationbound} to get
\begin{eqnarray}
&&\Pr[(u^*,v^*)\in \mathcal{K}\wedge \mathcal{E}_1]\nonumber\\
&\leq& \Pr[\exists (u,v)\in \mathcal{K}: \wht{\mathrm{Pr}}[\mathcal{R},u\sim v]\geq (1-\epsilon)(1+\epsilon_3)\mathrm{OPT}_k^c] \nonumber\\
&\leq& \sum\limits_{(u,v)\in \mathcal{K}} \Pr[\wht{\mathrm{Pr}}[\mathcal{R},u\sim v]\geq \mathrm{Pr}[u\sim v] +\epsilon_3(1-\epsilon)\mathrm{OPT}_k^c]\nonumber\\
&\leq& {n\choose 2}\exp\left(-\frac{3\epsilon_3^2(1-\epsilon)\mathrm{OPT}_k^c}{2(1+\epsilon_3)}|\mathcal{R}|\right)\leq \frac{n(n-1)}{n^2+n-2k}\delta\nonumber
\end{eqnarray}

Combining the above results, we get
\begin{eqnarray}
&&\Pr[\Pr[u^*\sim v^*]<(1-\epsilon)\mathrm{OPT}_k^c] \leq\Pr[(u^*,v^*)\in \mathcal{K}]\nonumber\\
&\leq& \Pr[(u^*,v^*)\in \mathcal{K}\wedge \mathcal{E}_1]+\Pr[\neg\mathcal{E}_1]\nonumber\\
&\leq& \delta_1+\delta_2\leq \delta
\end{eqnarray}
%
%Note that there are at most ${n\choose 2}$ different cluster-links. By using the union bound and the similar reasoning with that in Lemma~\ref{}, we must have $\mathrm{Pr}[\neg \mathcal{E}_1]\leq \delta/2$ and $\mathrm{Pr}[\neg \mathcal{E}_2]\leq \delta/2$ when $|\mathcal{R}|\geq \frac{4(\epsilon_3+2)}{3\epsilon_3^2}\ln\frac{n(n-1)}{\delta}$. Hence, the lemma follows.
Hence, the theorem follows.
\end{proof}

\subsection{Proof of Lemma~\ref{lma:judgegreedyinrandomworld1}}
\begin{proof}
Let $\wht{B}^o$ denote the $k$-clustering in $G$ such that $\wht{KC}(\R,\wht{B}^o)$ is maximized. If $\wht{KC}(\R,\wht{B}^o)>(1-\epsilon_3)q$, then we can get $\wht{L}_{\R}(q,C)\geq (n-\epsilon_1)(1-\epsilon_3)q$ by similar reasoning with that in Lemma~\ref{lma:judgegreedy}. Therefore, we get%Note that $\wht{KC}(\R,B^o_1)\geq \wht{KC}(\R,B^o)$. So we have
\begin{eqnarray}
&&\Pr[\wht{L}_{\R}(q,C)< (n-\epsilon_1)(1-\epsilon_3)q]\nonumber\\
&\leq& \Pr[\wht{KC}(\R,\wht{B}^o)\leq (1-\epsilon_3)q] \nonumber\\
&\leq& \Pr[\wht{KC}(\R,B^o)\leq (1-\epsilon_3)q] \nonumber\\
&\leq& \Pr[\exists (u,v)\in B^o: \wht{\Pr}[\R,u\sim v]\leq (1-\epsilon_3)q] \nonumber\\
&\leq& \sum\nolimits_{(u,v)\in B^o\wedge u\neq v}\Pr[\wht{\Pr}[\R,u\sim v]\leq (1-\epsilon_3)q] \nonumber\\
&\leq& (n-k)\exp\left(-\frac{3\epsilon_3^2 \OPT_k^c}{2(1+\epsilon_3)}|\mathcal{R}|\right)\leq \delta \nonumber
\end{eqnarray}
%$B$ denote a $k$-clustering in $G$ such that $\wht{\KC}(\R, B)$ is maximized. When $|C^*_q|> \lceil\ln\frac{n}{\epsilon_1 } \rceil k$, we must have $\wht{\KC}(\mathcal{R}_q,B^o_q)< (1-\epsilon_3)q$. Besides, we have $\wht{\KC}(\mathcal{R}_q,B^o)\leq \wht{\KC}(\mathcal{R}_q,B^o_q)$. So we can get
%\begin{eqnarray}
%&&\Pr[|C^*_q|>\lceil\ln\frac{n}{\epsilon_1 } \rceil k\wedge \mathrm{OPT}_k^c\geq q]\nonumber\\
%&\leq& \Pr[\wht{\KC}(\mathcal{R}_q,B^o)< (1-\epsilon_3)\mathrm{OPT}_k^c \wedge \mathrm{OPT}_k^c\geq q]\nonumber\\
%&\leq& \exp\left(-\frac{3\epsilon_3^2q}{2(1+\epsilon_3)}|\mathcal{R}_q|\right)\leq \delta_q
%\end{eqnarray}
Hence, the lemma follows.
\end{proof}

\subsection{Proof of Theorem~\ref{thm:expectednumberofsamples}}
\begin{proof}
Suppose that $q_{min}$ is the smallest $q$ that is tested by $\mathsf{SearchKC+}$. Then we must have $q_{min}\in \{2^{-j}|j\geq 1\}$. Suppose that $q_{min}=2^{-i_{min}}$. Then the $\mathsf{SearchKC+}$ algorithm takes $i_{min}$ iterations to reduce $q$ from $\frac{1}{2}$ to $q_{min}$, and then takes at most another $i_{min}+\lceil \log_2\frac{1}{\epsilon_2}\rceil$ iterations to end. Note that $i_{min}$ is a random number. When $i_{min}=i$, the total number of generated random samples is no more than
\begin{eqnarray}
\ell(i)= \left\lceil\frac{2^{i+1}(1+\epsilon_3)}{3\epsilon_3^2 (1-\epsilon)}\ln\frac{\pi^2(n^2+n-2k)(2i+\lceil \log_2\frac{1}{\epsilon_2}\rceil)^2}{12\delta}\right\rceil\nonumber
\end{eqnarray}
It can be verified that $\forall i\geq 1: \ell_{i+1}\leq 3\ell_i$. Suppose that $i_0\geq 1$ is the smallest number such that $2^{-i_0}\leq \OPT_k^c$. So we must have $2^{-i_0+1}\geq  \OPT_k^c$. For any $i>i_0$, we have
\begin{eqnarray}
&&\Pr[i_{min}=i]\leq \Pr[|C^*_{i-1}|>\lceil\ln({n}/{\epsilon_1 })\rceil k]\nonumber\\
&\leq& \Pr[\wht{\KC}(\mathcal{R}_{i-1},B^o)< (1-\epsilon_3)\mathrm{OPT}_k^c ]\nonumber\\
&\leq& \exp\left(-\frac{3\epsilon_3^2\OPT_k^c}{2(1+\epsilon_3)}|\mathcal{R}_{i-1}|\right)\leq \delta^{\OPT_k^c\cdot 2^{i-1}}
\end{eqnarray}
where $\R_{i-1}$ and $C^*_{i-1}$ denote the set of generated random samples and the set of centering nodes found by $\mathsf{SearchKC+}$ when $q=2^{-i+1}$, respectively.
Let $a=\OPT_k^c\cdot 2^{i_0}$. When $\delta\leq 1/3$,  the total expected number of generated random samples is no more than
\begin{eqnarray}
&&\ell(i_0)+\sum_{i>i_0} \ell(i)\Pr[i_{min}=i]\nonumber\\
&\leq&\ell(i_0)+\sum_{j=0}^\infty 3^{j+1}\ell(i_0)\delta^{2^j a}\leq 2\ell(i_0)+ 3\ell(i_0)\sum_{j=1}^\infty 3^{j-2^j}\nonumber\\
&\leq&2\ell(i_0)+3\ell(i_0)\sum_{j=1}^\infty 3^{-j}\leq 7\ell(i_0)/2 \nonumber
\end{eqnarray}
Note that $i_0=\lceil \log_2\frac{1}{\OPT_k^c}\rceil$, so we have
\begin{eqnarray}
\ell(i_0)=\mathcal{O}\left(\frac{1}{\epsilon^2 \OPT_k^c}\left(\ln \frac{n}{\delta}+\ln\ln \frac{1}{\epsilon\OPT_k^c}\right)\right)
\end{eqnarray}
Hence the theorem follows.
\end{proof}

\end{appendix}

\bibliographystyle{ACM-Reference-Format}
\bibliography{mylib}


\begin{thebibliography}{8}

%%% the Web address.

\ifx \showCODEN    \undefined \def \showCODEN     #1{\unskip}     \fi
\ifx \showDOI      \undefined \def \showDOI       #1{#1}\fi
\ifx \showISBNx    \undefined \def \showISBNx     #1{\unskip}     \fi
\ifx \showISBNxiii \undefined \def \showISBNxiii  #1{\unskip}     \fi
\ifx \showISSN     \undefined \def \showISSN      #1{\unskip}     \fi
\ifx \showLCCN     \undefined \def \showLCCN      #1{\unskip}     \fi
\ifx \shownote     \undefined \def \shownote      #1{#1}          \fi
\ifx \showarticletitle \undefined \def \showarticletitle #1{#1}   \fi
\ifx \showURL      \undefined \def \showURL       {\relax}        \fi
% invisible to TeX
\providecommand\bibfield[2]{#2}
\providecommand\bibinfo[2]{#2}
\providecommand\natexlab[1]{#1}
\providecommand\showeprint[2][]{arXiv:#2}

\bibitem[\protect\citeauthoryear{Ceccarello, Fantozzi, Pietracaprina, Pucci,
  and Vandin}{Ceccarello et~al\mbox{.}}{2017}]%
        {Ceccarello17}
\bibfield{author}{\bibinfo{person}{Matteo Ceccarello}, \bibinfo{person}{Carlo
  Fantozzi}, \bibinfo{person}{Andrea Pietracaprina}, \bibinfo{person}{Geppino
  Pucci}, {and} \bibinfo{person}{Fabio Vandin}.}
  \bibinfo{year}{2017}\natexlab{}.
\newblock \showarticletitle{Clustering Uncertain Graphs}.
\newblock \bibinfo{journal}{\emph{{PVLDB}}} \bibinfo{volume}{11},
  \bibinfo{number}{4} (\bibinfo{year}{2017}), \bibinfo{pages}{472--484}.
\newblock


\bibitem[\protect\citeauthoryear{Feige}{Feige}{1998}]%
        {feige1998threshold}
\bibfield{author}{\bibinfo{person}{Uriel Feige}.}
  \bibinfo{year}{1998}\natexlab{}.
\newblock \showarticletitle{A threshold of ln n for approximating set cover}.
\newblock \bibinfo{journal}{\emph{J. ACM}} \bibinfo{volume}{45},
  \bibinfo{number}{4} (\bibinfo{year}{1998}), \bibinfo{pages}{634--652}.
\newblock


\bibitem[\protect\citeauthoryear{Goyal, Bonchi, Lakshmanan, and
  Venkatasubramanian}{Goyal et~al\mbox{.}}{2013}]%
        {Goyal2013}
\bibfield{author}{\bibinfo{person}{Amit Goyal}, \bibinfo{person}{Francesco
  Bonchi}, \bibinfo{person}{Laks V.~S. Lakshmanan}, {and}
  \bibinfo{person}{Suresh Venkatasubramanian}.}
  \bibinfo{year}{2013}\natexlab{}.
\newblock \showarticletitle{On minimizing budget and time in influence
  propagation over social networks}.
\newblock \bibinfo{journal}{\emph{Social Network Analysis and Mining}}
  \bibinfo{volume}{3}, \bibinfo{number}{2} (\bibinfo{year}{2013}),
  \bibinfo{pages}{179--192}.
\newblock


\bibitem[\protect\citeauthoryear{Krause, McMahan, Guestrin, and Gupta}{Krause
  et~al\mbox{.}}{2008}]%
        {krause2008robust}
\bibfield{author}{\bibinfo{person}{Andreas Krause}, \bibinfo{person}{H~Brendan
  McMahan}, \bibinfo{person}{Carlos Guestrin}, {and} \bibinfo{person}{Anupam
  Gupta}.} \bibinfo{year}{2008}\natexlab{}.
\newblock \showarticletitle{Robust submodular observation selection}.
\newblock \bibinfo{journal}{\emph{Journal of Machine Learning Research}}
  \bibinfo{volume}{9}, \bibinfo{number}{Dec} (\bibinfo{year}{2008}),
  \bibinfo{pages}{2761--2801}.
\newblock


\bibitem[\protect\citeauthoryear{Leskovec, Krause, Guestrin, Faloutsos,
  VanBriesen, and Glance}{Leskovec et~al\mbox{.}}{2007}]%
        {LeskovecKGFVG2007}
\bibfield{author}{\bibinfo{person}{Jure Leskovec}, \bibinfo{person}{Andreas
  Krause}, \bibinfo{person}{Carlos Guestrin}, \bibinfo{person}{Christos
  Faloutsos}, \bibinfo{person}{Jeanne VanBriesen}, {and}
  \bibinfo{person}{Natalie Glance}.} \bibinfo{year}{2007}\natexlab{}.
\newblock \showarticletitle{Cost-effective outbreak detection in networks}. In
  \bibinfo{booktitle}{\emph{KDD}}. \bibinfo{pages}{420--429}.
\newblock


\bibitem[\protect\citeauthoryear{Tang, Tang, Xiao, and Yuan}{Tang
  et~al\mbox{.}}{2018}]%
        {Tang2018}
\bibfield{author}{\bibinfo{person}{Jing Tang}, \bibinfo{person}{Xueyan Tang},
  \bibinfo{person}{Xiaokui Xiao}, {and} \bibinfo{person}{Junsong Yuan}.}
  \bibinfo{year}{2018}\natexlab{}.
\newblock \showarticletitle{Online Processing Algorithms for Influence
  Maximization}. In \bibinfo{booktitle}{\emph{SIGMOD}}.
  \bibinfo{pages}{991--1005}.
\newblock


\bibitem[\protect\citeauthoryear{Vazirani}{Vazirani}{2001}]%
        {Vazirani2001}
\bibfield{author}{\bibinfo{person}{Vijay~V. Vazirani}.}
  \bibinfo{year}{2001}\natexlab{}.
\newblock \bibinfo{booktitle}{\emph{Approximation Algorithms}}.
\newblock \bibinfo{publisher}{Springer-Verlag}, \bibinfo{address}{Berlin}.
\newblock


\bibitem[\protect\citeauthoryear{Williamson and Shmoys}{Williamson and
  Shmoys}{2011}]%
        {williamson2011design}
\bibfield{author}{\bibinfo{person}{David~P Williamson} {and}
  \bibinfo{person}{David~B Shmoys}.} \bibinfo{year}{2011}\natexlab{}.
\newblock \bibinfo{booktitle}{\emph{The design of approximation algorithms}}.
\newblock \bibinfo{publisher}{Cambridge university press}.
\newblock


\end{thebibliography}

\end{document}